\documentclass[11pt]{article}

\usepackage{fullpage}
\usepackage[utf8]{inputenc} 
\usepackage[T1]{fontenc}    
\usepackage{hyperref}       
\usepackage{url}            
\usepackage{booktabs}       
\usepackage{amsfonts}       
\usepackage{nicefrac}       
\usepackage{microtype}      
\usepackage{graphicx}
\usepackage{xcolor}
\usepackage{soul}
\usepackage{doi}
\usepackage{authblk}

\usepackage{amsmath}
\usepackage{amssymb}
\usepackage{enumerate}
\usepackage{epstopdf}
\usepackage{algorithmic}
\usepackage{mathtools}
\usepackage{setspace}

\usepackage{amsthm}

\theoremstyle{plain}
\newtheorem{theorem}{Theorem}
\numberwithin{theorem}{section}
\newtheorem{proposition}{Proposition}
\numberwithin{proposition}{section}
\newtheorem{lemma}{Lemma}
\numberwithin{lemma}{section}

\numberwithin{corollary}{section}

\numberwithin{example}{section}

\theoremstyle{definition}
\newtheorem{definition}{Definition}
\numberwithin{definition}{section}
\newtheorem{remark}{Remark}
\numberwithin{remark}{section}
\newtheorem{notation}{Notation}
\numberwithin{notation}{section}

\numberwithin{equation}{section}

\providecommand{\keywords}[1]
{
  \small	
  \textbf{Keywords:} #1
}
\title{Uniqueness of MAP estimates for inverse problems under information field theory}

\date{May 26 2025}

\begin{document}

\author[1]{Alex Alberts\thanks{Corresponding author   \href{mailto:albert31@purdue.edu}{\color{blue}{\texttt{albert31@purdue.edu}}}}}
\author[1]{Ilias Bilionis}
\affil[1]{School of Mechanical Engineering, Purdue University, West Lafayette, IN}

\maketitle

\begin{abstract}
    Information field theory (IFT) is an emerging technique for posing infinite-dimensional inverse problems using the mathematics found in quantum field theory.
    Under IFT, the field inference task is formulated in a Bayesian setting where the probability measures are defined by path integrals.
    We derive conditions under which IFT inverse problems have unique maximum a posterioi (MAP) estimates, placing a special focus on the problem of identifying model-form error.
    We define physics-informed priors over fields, where a parameter, called the model trust, measures our belief in the physical model.
    Smaller values of trust cause the prior to diffuse, representing a larger degree of uncertainty about the physics.
    To detect model-form error, we learn the trust as part of the inverse problem and study the limiting behavior.
    We provide an example where the physics are assumed to be the Poisson equation and study the effect of model-form error on the model trust.
    We find that a correct model leads to infinite trust, and under model-form error, physics that are closer to the ground truth lead to larger values of the trust.
\end{abstract}


\onehalfspacing

\section{Introduction}
\label{sec:intro}

One of the most fundamental tasks within the realm of uncertainty quantification is the ability to unravel hidden variables from observed data.
Such \emph{inverse problems} are encountered across numerous disciplines within science and engineering.
Because most inverse problems are ill-posed in the classical sense of Hadamard~\cite{kabanikhin2008definitions}, methodologies which solve inverse problems are often posed in a Bayesian way.
In this type of approach, some assumed prior knowledge is placed on the hidden variables.
The prior information is then combined with the data through a likelihood function, and Bayes's theorem is applied to derive the posterior over the hidden variables.

While the treatment of inverse problems is well-studied, it remains a nuanced task with multiple challenges.
Some of these include the high-dimensionality of the problem, computational complexity, the acquisition of expensive datasets, and the selection of prior information, among others.
In this paper, we place a specific focus on the problem of well-posedness in the solution of inverse problems through the lens of the maximum a posterioi (MAP) estimate.
Many inverse problems have multiple candidate solutions, and identifying the correct solution becomes a critical challenge.
This situation commonly appears when there is insufficient data, e.g., missing boundary conditions, unobserved quantities, etc., and ill-posedness may manifest in other ways.
See~\cite{ghattas2021learning} for an in-depth discussion on ill-posed inverse problems.

We study inverse problems through the lens of information field theory (IFT), which is a methodology for performing Bayesian inference over fields by making use of probability measures defined over function spaces using the path integral formalism~\cite{ensslin2009information}.
A field is a physical quantity with a value for every point in space and time, including scalar, vector, or tensor fields.
The mathematics of IFT are based in the theory of path (functional) integration~\cite{cartier2006functional}, and we can rely on various properties of the integral to derive some useful results.
IFT was originally developed for applications in cosmology~\cite{elsner2010local, oppermann2012improved, westerkamp2024first, edenhofer2024parsec}, and has seen some interest in medical imaging~\cite{frank2016dynamic, galinsky2017unified, frank2020jedi}, computational biology~\cite{guardiani2021non, chen2021field}, and engineering ~\cite{pandey2022review}.
Furthermore, it can be shown that other approaches such as physics-informed neural networks (PINNs)~\cite{raissi2019physics}, Bayesian-PINNs~\cite{yang2021b}, or even Gaussian process regression~\cite{schulz2018tutorial} are limiting cases of IFT, and results derived from IFT can naturally be extended to these methodologies as well.

In this work, we relax the definition of a well-posed inverse problem to one where the MAP estimate exists and is unique.
We refer to such problems as \emph{weakly well-posed}.
The posterior may also appear to be flat, in which case the problem does not have an identifiable solution.
Simply put, between the selection of the prior and the observed data, there may not be enough information available to identify a unique `best guess' of the hidden variables.
Determining when an inverse problem has a unique MAP estimate can prevent needlessly wasting resources collecting additional data, which can be costly or even impossible in many cases.

We define what it means for an inverse problem to be weakly well-posed in the context of IFT, and prove a few useful theorems for conditions which are sufficient to make an inverse problem well-posed.
Some analytical results are derived, with a special focus placed on the problem of detecting the presence of model-form error.
The paper is organized as follows: a brief overview of IFT is provided in Sec.~\ref{sec:iftreview}, along with a discussion of the application of IFT to the problem of detecting model-form error. 
In Sec.~\ref{sec:inverseprobs}, we explore different properties of MAP estimates under IFT, and derive a theorem for identify when such a MAP estimate is unique.
Finally, in Sec.~\ref{sec:example} we provide a brief example of how the theorems can be used in practice, with a focus on model-form error identification.

\section{Review of information field theory}
\label{sec:iftreview}

IFT is a Bayesian methodology for quantifying uncertainty over fields, and the various mathematical objects it studies are derived using the path integral formalism, which commonly appear in statistical field theory~\cite{parisi1988statistical} and in quantum field theory~\cite{lancaster2014quantum}.
IFT begins by defining a prior probability measure over fields, which typically encodes some physical knowledge about the field, e.g., regularity, smoothness, symmetries, differential equations which the field is known to obey, and others.
Combining this prior with data in the form of the likelihood via Bayes's theorem provides the posterior knowledge about the field(s) of interest.

We briefly review the traditional Bayesian approach to inverse problems~\cite{stuart2010inverse} so that we may communicate clearly how IFT treats inverse problems differently.
As an example, consider a model of a physical process, which we take to be a partial differential equation (PDE),
\begin{equation}
    \label{eqn:pde}
    \left\{
    \begin{split}
        P[\phi;u] &= 0 \quad \mathrm{on} \quad \Omega \subset \mathbb{R}^d \\
        B[\phi] &= 0 \quad \mathrm{on} \quad \partial \Omega,
    \end{split}
    \right.
\end{equation}
where $P$ is a differential operator, $B$ prescribes the boundary conditions, $\phi$ is the solution to the PDE, and $u$ is some unknown input field.
Under the classical approach, the structure of the PDE is assumed to be perfectly correct in the sense that once $u$ is recovered, we may determine $\phi$ by solving the PDE (the forward problem).
The inverse problem is then simply to identify the input $u$.
This is usually performed by assuming a Gaussian process prior for $u$, while the solution of the PDE on collocation points defines the likelihood.
In contrast, under IFT it is typical to relax the assumption that the model is correct.
We do so by encoding the physics into the prior, rather than the likelihood.
Then, the inverse problem involves inferring both $\phi$ and $u$.

\subsection{Setting and notation}

In practice, we seek to infer a field of interest $\phi : \Omega \to \mathbb{R}$, where $\Omega \subseteq \mathbb{R}^{d_{\Omega}}$, and $\phi$ belongs to some Banach space $\mathcal{X}$.
We will let $\phi^*$ denote the unknown ground truth field, which is distinct from any estimate $\phi$.
Under our treatment of IFT, we are primarily concerned with inferring fields which are thought to satisfy some PDE, hence we will often take $\mathcal{X}$ to be a Sobolev space or otherwise some subset of $L^2(\Omega)$.
Given a multi-index $\alpha = (\alpha_1 \dots \alpha_q) \in \mathbb{N}^q$ with order $|\alpha| = \sum_i \alpha_i$, and a field $\phi \in \mathcal{X}$, let
$$
D^{\alpha}\phi = \frac{\partial^{|\alpha|}}{\partial x_1^{\alpha_1}\cdots\partial x_n^{\alpha_q}}\phi
$$
represent the mixed (weak) partial derivative of $\phi$.
Given $\tau \in \mathbb{N}$, denote the Sobolev space of functions on $\Omega$ with square-integrable weak derivatives up to order $\tau$ by $W^{\tau}$, that is,
$$
W^{\tau} = \left\{\phi \in L^2(\Omega) : D^{\alpha}\phi \in L^2(\Omega) \:\forall\: |\alpha| \leq \tau \right\}.
$$
We also equip the fields with the inner product on $L^2(\Omega)$ denoted by
$$
\psi^{\dagger}\phi \coloneqq \langle \phi, \psi \rangle_{L^2(\Omega)} = \int_{\Omega} \phi \psi \:d\Omega,
$$
for $\phi,\psi \in L^2(\Omega)$.
In the case of an operator $F : \mathbb{R} \to \mathcal{X}$, $F^*$ denotes the adjoint of the operator, defined by the unique map $F^*: \mathcal{X} \to \mathbb{R}$ such that
$$
cF(\phi) = \langle \phi, F^*(c)\rangle_{L^2(\Omega)},
$$
for $\phi \in \mathcal{X}$, $c\in\mathbb{R}$.

\subsection{Construction of information field theory posteriors}

IFT begins by defining a prior probability measure over the space of field configurations, formally denoted by $p(\phi)$.
To construct this prior, we follow the approach as outlined in~\cite{alberts2023physics}, and later extended to dynamical systems in~\cite{hao2023information}.
This approach makes use of physics-informed priors coming from some field energy functional, the minimization of which provides the state of the system.
The energy functional usually is defined from a PDE, which aligns with physics-informed techniques under the scientific machine learning paradigm~\cite{karniadakis2021physics}.
Alternatively, one could consider an action principle or Tikhonov regularization~\cite{chang2014path}.
We construct the theorems here so that they apply to IFT in general, and the physics-informed case serves as a motivating example.
In classical applications of IFT,~\cite{ensslin2009information}, the prior is typically taken to be a zero-mean Gaussian random field, where the covariance kernel is selected to match some field regularity constraints.
Oftentimes, the correlation structure of the field is learned from the data rather than assumed, and our theorems on parameter identification stated here may apply to this case by relaxing some assumptions.

We assume that the current state of knowledge of the physics is captured by a generic local energy functional
\begin{equation}
    E(\phi;u) = \int_{\Omega} h\left(x, \phi(x), \nabla_x\phi(x),\dots ;u \right)dx,
    \label{eqn:energy}
\end{equation}
where $h$ is an energy density function, and $u$ represents various parameters or additional fields, such as a source term.
The minimization of eq.~(\ref{eqn:energy}), constrained on any boundary conditions, yields the desired field $\phi$.
The task is then to infer $\phi$ or the parameters (or additional fields) $u$, the regime of classical Bayesian inverse problems.
Typically, the energy is derived from a boundary value problem described by the PDE eq.~(\ref{eqn:pde}).
In some cases, a variational principle for eq.~(\ref{eqn:pde}) exists, and an exact field energy can be derived.
A classic example of this is Dirichlet's principle, which describes the variational form of the Poisson equation~\cite{brezis2011functional}.
Otherwise, the integrated square residual of the PDE can be used in applications as an approximation to the field energy.

IFT makes use of the path integral formalism, where we define a Boltzmann-like physics-informed prior for the field conditional on the parameters:
\begin{equation}
    p\left(\phi | u, \beta\right) = \frac{\exp\left\{-H\left(\phi;u, \beta\right)\right\}}{Z\left(u, \beta\right)},
    \label{eqn:piftprior}
\end{equation}
where 
$$
H(\phi ; u, \beta) = \beta E(\phi;u)
$$
is a functional of the field known as the \emph{prior potential}, and
\begin{equation}
    \label{eqn:pathintegral}
    Z(u, \beta) = \int \mathcal{D}\phi\: \exp\left\{-H(\phi ; u, \beta)\right\}
\end{equation}
is the \emph{partition function}, which acts as the normalization constant.
The formally defined probability measure of eq.~(\ref{eqn:piftprior}) is defined over the space of field configurations, $\mathcal{X}$, and the integral contained in $Z(u,\beta)$ is a taken to be a path integral.
Note that in IFT literature, the above objects are often posed in an information theoretical way, i.e., we define an \emph{information Hamiltonian} to be $-\log p(\phi|u,\beta)$.
Our prior potential is equivalent to a typical information Hamiltonian, up to the normalization constant.
We are choosing to separate the normalization constant as this is needed for some of the proofs in this work.

For forward problems where we only seek to infer the field, the partition function may be treated as a constant (although it may be infinite, or even ill-defined), and dropped for sampling algorithms.
For potentials which have a quadratic and positive-definite form, that is, the potential can be written as
$$
H(\phi;u) = \frac{1}{2}\phi^{\dagger}C^{-1}\phi,
$$
where $C:\mathcal{X} \to \mathcal{X}$ is trace class and positive-definite, then eq.~(\ref{eqn:piftprior}) can be given a rigorous treatment.
This is done by defining the prior on a finite-dimensional subspace of $\mathcal{X}$ and recovering a Gaussian random field in the limit~\cite{glimm2012quantum}.
There are even certain PDEs for which this holds.
For example, in~\cite{alberts2025interpretation}, it is shown that for the Poisson equation, eq.~(\ref{eqn:piftprior}) is related to the Brownian bridge measure.
The physics-informed prior given by eq.~(\ref{eqn:piftprior}) is intuitively defined in such a way where fields which are closer to minimizing eq.~(\ref{eqn:energy}) are a priori more likely.

In addition to the physics, a nonnegative scaling parameter, $0< \beta <\infty$, is included in the prior to ensure that the potential remains unitless.
It is analogous to the inverse-temperature parameter of the Boltzmann distribution and controls the spread of the prior.
The measure collapses to a minimum field energy as $\beta$ tends to infinity, and the measure becomes flat as $\beta \to 0$.
In this way we find that $\beta$ is a pseudo-metric which quantifies our belief in the chosen physical knowledge.
That is, $\beta$ quantifies with how much certainty we believe eq.~(\ref{eqn:energy}) to represent the true underlying physics which the field is assumed to obey.
For this reason, $\beta$ is referred to as the \emph{model trust} parameter.

From this point forward, to simplify notation, we consolidate the parameters into the tuple $\lambda = (u,\beta)$.
For inverse problems, we define a prior over various hyperparameters/hyperfields via $p(\lambda)$, chosen according to the application at hand.
For example, if $u$ is a field, it may be assigned a Gaussian process prior as usual.
We will see certain assumptions about the structure of this prior must be made to satisfy various theorems derived here.
The model trust $\beta$ is typically assigned a flat prior or Jeffrey's prior over $\mathbb{R}_+$.
We will define the quantity $H(\lambda)= -\log p(\lambda)$ to be the \emph{parameter potential}.

As in any application of Bayesian inference, the next step is to define the likelihood.
We will restrict ourselves to linear measurements with additive Gaussian noise.
That is, suppose we collect data $d \in \mathbb{R}^n$ following 
\begin{equation}
    \label{eqn:data}
    d = R\phi + \gamma,
\end{equation}
where $R : \mathcal{X} \to \mathbb{R}^n$ is the measurement operator, and $\gamma \sim \mathcal{N}(0, \Gamma)$ is the noise process with an $n\times n$ noise covariance matrix $\Gamma$.
Hence, the likelihood is given by $p(d | \phi) = \mathcal{N}(d|R\phi, \Gamma)$, and we define the \emph{data potential} to be $H(d;\phi) = - \log p(d|\phi)$.

The IFT joint posterior over the field and parameters is derived by formal application of Bayes's theorem: $p(\phi,\lambda | d) = \frac{1}{p(d)}p(d|\phi)p(\phi|\lambda)p(\lambda)$.
All information about the posterior can be summarized by the joint \emph{posterior potential}
\begin{equation}
\label{eqn:joint}
    H(\phi, \lambda ; d) \:\hat{=}\: H(d;\phi) + H(\phi;\lambda) - \log  Z(\lambda)+ H(\lambda)
\end{equation}
and corresponding normalization constant\footnote{We have dropped the normalization constant of the likelihood from this expression, as it does not depend on $\phi$ nor on $\lambda$. For this reason we have written `$\hat{=}$.'}
\begin{equation*}
\label{eqn:postpartition}
    p(d) = \int \mathcal{D}\phi\: \int d\lambda \: \exp\left\{-H(\phi,\lambda;d)\right\}.
\end{equation*}
Under special cases, the IFT posterior derived from a \emph{forward} problem can be characterized rigorously as a Gaussian random field, referred to as \emph{free theory}~\cite{ensslin2019information}.
\begin{definition}[Free theory]
    \label{def:freetheory}
    Suppose the conditional prior $p(\phi|\lambda)$ can be characterized as a zero-mean Gaussian random field with covariance operator $S_{\lambda}$, i.e., $p(\phi|\lambda) = \mathcal{N}(\phi|0, S_{\lambda})$.
    Further suppose that the data is collected according to eq.~(\ref{eqn:data}). 
    Then, we say we work under free theory, and the posterior is also a Gaussian random field, given by $p(\phi | d, \lambda) = \mathcal{N}(\phi|\Tilde{m}_{\lambda}, \Tilde{S}_{\lambda})$.
    The expressions for the posterior mean and covariance are well known~\cite{seeger2004gaussian}, namely:
    $$
    \begin{array}{ll}
        \Tilde{m}_{\lambda} =(S^{-1}_{\lambda}+R^{*}\Gamma^{-1}R)^{-1}R^{*}\Gamma^{-1}d  \\
         \Tilde{S}_{\lambda} = (S^{-1}_{\lambda}+R^{*}\Gamma^{-1}R)^{-1}. 
    \end{array}
    $$
\end{definition}

The partition function greatly complicates characterizing the full posterior.
For forward problems where $\lambda$ is kept fixed, eq.~(\ref{eqn:pathintegral}) can be treated as a constant, and eq.~(\ref{eqn:joint}) summarizes all posterior knowledge.
The posterior can then be characterized numerically in a number of means, for example, stochastic gradient Langevin dynamics (SGLD) and variants~\cite{alberts2023physics, hao2025neural} or stochastic variational inference~\cite{hao2023information, knollmuller2019metric, frank2021geometric}.
We are interested in studying classical Bayesian inverse problems, where $\lambda$ is learned from data.
In this situation, the partition function is no longer constant and appears in the joint posterior.
Because the partition function is defined as a path integral, this presents various mathematical issues which must be treated with care.
We derive some properties of inverse problems in a way which avoids the issue of characterizing the partition function.
First, we briefly discuss the problem of detecting model-form error, an important inverse problem which can be studied through IFT.

\subsection{Discussion on model-form error}
\label{sec:modelerrordisc}
From here we remark how the model trust parameter can be used in the application of detecting model-form error.
Typically a series of assumptions are made to derive the physical model, and the model is simply an approximation of the ground truth physics.
Model-form error is the epistemic uncertainty which arises from an imperfect physical model, identified as a critical challenge over two decades ago~\cite{draper1995assessment, kennedy2001bayesian}.
The recent surge in popularity of physics-informed models has led to a new generation of methods for quantifying model-form error~\cite{sargsyan2019embedded, park2014bayesian, riley2011quantification, geneva2019quantifying, brynjarsdottir2014learning}.
The available methods typically couple a solver for the physics with a technique from uncertainty quantification.
This involves resolving the problem multiple times, creating a bottleneck.

The physics-informed IFT approach avoids this issue by including the trust directly in the prior over the field.
Recall that the prior potential conditional on the parameters is the field energy scaled by a factor $\beta$.
This implies that $\beta$ controls the strength of the contribution of the physics to the posterior.
The case $\beta=0$ corresponds to the selection of a flat field prior, and the physics plays no role in the structure of the posterior.
There is no trust in the physics.
The posterior behavior is dominated by the likelihood, and the method reduces to a Bayesian regression scheme with a flat prior placed on the field.
As the trust begins to increase, the physics contributes a greater effect on the structure of the posterior, eventually dominating the posterior behavior altogether.
The extreme case $\beta\rightarrow\infty$ arises when complete belief is placed in the physics.
The prior collapses to the field which minimizes the energy, and the only field considered is the one a priori assumed to be correct.
This is a manifestation of the belief that the model is perfect.
If the field energy is the variational form of a PDE, then this minimum field energy is exactly the solution to the boundary value problem, if the field is a priori assumed to satisfy the boundary conditions (before we have observations in the domain).

We are primarily concerned with the case where $0<\beta<\infty$.
That is, when we have some moderate trust in the physics, which leads to finite, nonzero variance in the sample fields from the prior.
Many physical systems of interest operate in this regime, as the physics selected to model the system are not a perfect representation of the ground truth which generates the data.
Likewise, this is a natural treatment if eq.~(\ref{eqn:energy}) is the integrated square residual of a PDE, which serves only as an approximation to the physics.

We observe that the selection of the trust has a direct effect on the variance of samples from the physics-informed prior.
In~\cite{alberts2023physics,hao2023information}, this effect is exploited to study the problem of detecting model-form error in numerical examples, and an analytical example for the Poisson equation can be found in~\cite{alberts2025interpretation}.
By inferring the trust as a hyperparameter in an inverse problem, we allow the model to automatically select the belief in the physics by scaling this variance.
In this way, the model calibrates the balance between the empirical data and the physics.
If the posterior identifies a low trust in the physics, this should serve as a flag that the model needs to be corrected.
Analytical examples show that the trust scales the prior covariance by $\beta^{-1}$ in the free theory, demonstrating this behavior.
Similarly, numerical experiments provide evidence that as the physics strays further from the truth, the model selects a lower trust, and the prior becomes flat.
On the other hand, as the physics becomes closer to reality, the trust grows, and the posterior begins to collapse to a single field.
In this paper, we study the existence and uniqueness of MAP estimates to inverse problems, under relatively light conditions.
This can readily be applied to the inverse problem of detecting model-form error.

\section{Some properties of MAP estimates of IFT inverse problems}
\label{sec:inverseprobs}

We derive some results related to MAP estimates of inverse problems under the IFT paradigm.
In particular, we define the concept of a weakly well-posed inverse problem under IFT and we state a few useful theorems for when an inverse problem becomes weakly well-posed.
Note that the results derived here extend to any physical parameters contained in eq.~(\ref{eqn:energy}) as well as the trust.
Unknown fields contained in the physics can also be inferred by first expressing the field as a linear combination of basis functions, where we look to identify the parameters of the basis, e.g., the Karhunen-Lo\`eve expansion of a Gaussian process.
We provide some useful results related to inverse problems where we infer arbitrary parameters. 

\subsection{Preliminaries}

The definition of an inverse problem in the IFT context is given by the state of knowledge about the hidden parameters $\lambda$ after observing the data.
This can be found by taking the marginal of the joint posterior $p(\phi,\lambda|d)$.
Integrating out the field, we obtain
\begin{equation}
\label{eqn:lambda_marginal}
p(\lambda|d) = \int\mathcal{D}\phi\: p(\phi,\lambda|d) = \int\mathcal{D}\phi\: \frac{1}{p(d)}p(d|\phi)\frac{\exp\left\{-\left[H(\phi;\lambda)+H(\lambda)\right]\right\}}{Z(\lambda)}.
\end{equation}
Interpreting eq.~(\ref{eqn:lambda_marginal}) is not trivial due to the path integrals which may not be tractable outside of free theory.
However, we will see this is not an issue as the theorems derived here only depend on the derivatives of the \emph{parameter posterior potential} $H(\lambda;d) \coloneqq -\log p(\lambda|d)$, for which there are readily available formulae.

As with any inverse problem, we may commonly find ourselves in the case where the problem is not well-posed in the deterministic sense of Hadamard~\cite{hadamard1902problemes}.
An inverse problem is well-posed if the solution exists, is unique, and is stable with respect to the data.
This definition is not particularly useful in the context of Bayesian inference, as the posterior having any amount of uncertainty causes the inverse problem to be ill-posed.
Therefore, for a Bayesian inverse problem, well-posedness is characterized by the existence, uniqueness, and stability of the \emph{posterior} with respect to a probabilistic metric, e.g., the Hellinger distance~\cite{stuart2010inverse}.
This definition is fairly loose and is satisfied under most Bayesian inverse problems~\cite{latz2020well}.
Also the well-posedness of a Bayesian inverse problem depends on the specific choice of the prior.
One could therefore imagine a transformation of the prior where the inverse problem is no longer Bayesian well-posed.
When we work under IFT, we explicitly make the assumption that the inverse problem is well-posed in the Bayesian sense.
For these reasons, we relax the definition and define the concept of a \emph{weakly} well-posed inverse problem, which we simply take to be the posterior MAP estimate of a given posterior exists and is unique.

To use IFT in theoretical investigations, we derive theorems to provide conditions for which a given inverse problem is weakly well-posed.
Precisely, we have:
\begin{definition}[Weakly well-posed inverse problem]
\label{def:well-posed}
An inverse problem under IFT is weakly well-posed if the parameter posterior potential $H(\lambda;d)$ has a unique minimum.
This statement is equivalent to the marginal posterior over the parameters $p(\lambda|d)$ being unimodal.
\end{definition}

\begin{remark}
    Our definition of a well-posed inverse problem is related to the concept of the identifiability of parameters~\cite{ying1991asymptotic, zhang2004inconsistent, daly2018inference}.
    To summarize, a parameter is identifiable if we can uniquely learn the true value of said parameter in the limit of infinite data, under the assumption that the model is correct.
    Def.~\ref{def:well-posed} does not require infinite data or convergence to the ground truth parameters.
    Further, we allow for the case of an imperfect model.
    In fact, as elaborated on in Sec.~\ref{sec:modelerrordisc}, the problem of identifying when the model is incorrect can be posed as an inverse problem under IFT.
    We show a case where this problem becomes weakly well-posed in Sec.~\ref{sec:example}.
\end{remark}

We can now prove the following:

\begin{proposition}
    \label{lem:lemma1}
    Suppose that the Hessian of the parameter posterior potential, given by $\nabla_{\lambda}^2H(\lambda;d)$, is positive definite.
    If there exists $\lambda^*$ such that $\nabla_{\lambda}H(\lambda^*;d) = 0$, then the inverse problem is weakly well-posed and $\lambda^*$ is the unique MAP estimate.
\end{proposition}

\begin{proof}
    The proof uses trivial facts from the theory of convex functions.
    Since the Hessian is positive definite, the potential is strictly convex.
    This implies that the $\lambda^*$ which makes $\nabla_{\lambda}H(\lambda^*;d) = 0$ is the unique global minimizer.
    Hence, Def.~\ref{def:well-posed} applies, and the inverse problem is weakly well-posed.
\end{proof}

So, to study the well-posedness of inverse problems we need to characterize both the gradient $\nabla_{\lambda}H(\lambda;d)$ and the Hessian $\nabla^2_{\lambda}H(\lambda;d)$ of the parameter posterior potential $H(\lambda;d)$.
Expressions for these can be found using properties of expectations.
To make progress, we must first state a couple of definitions.
\begin{definition}[Conditional field expectations]
    \label{def:expect}
    Let $F : \mathcal{X} \to \mathbb{R}$ be a functional, and let $p(\phi | \lambda)$ and $p(\phi | d, \lambda)$ denote the field prior and posterior probability measure with fixed $\lambda$, respectively.
    The conditional expectation of $F$ over the field prior is given by the path integral
    $$
    \mathbb{E}\left[F(\phi)\middle|\lambda\right] \coloneqq \int\mathcal{D}\phi\:F(\phi)p(\phi|\lambda).
    $$
    Further, the conditional expectation of $F$ over the field posterior is
    $$
    \mathbb{E}\left[F(\phi)\middle|d,\lambda \right] \coloneqq \int\mathcal{D}\phi\: F(\phi)p(\phi|d,\lambda).
    $$
\end{definition}

In~\cite{alberts2023physics}, an expression for $\nabla_{\lambda}H(\lambda|d)$ is derived using the above expectations.
We restate the proof so that our work is self-contained and because an identity found in the proof is needed to derive the Hessian.
\begin{lemma}
    \label{eqn:unbiased_grad_lambda}
    The gradient of the parameter posterior potential is:
    $$
    \nabla_{\lambda}H(\lambda;d) = \mathbb{E}\left[\nabla_{\lambda}H(\phi;\lambda)\middle| d,\lambda\right] - \mathbb{E}\left[\nabla_{\lambda}H(\phi;\lambda)\middle| \lambda \right]  + \nabla_{\lambda} H(\lambda).
    $$
\end{lemma}
\begin{proof}
    We find the gradient element-wise.
    To begin,
    $$
    \frac{\partial H(\lambda;d)}{\partial\lambda_i} =  -\frac{\partial }{\partial \lambda_i}\log p(\lambda|d) = -\frac{\partial }{\partial \lambda_i} \log\int\mathcal{D}\phi\: p(\phi,\lambda|d),
    $$
    and by the chain rule, we have
    \begin{equation}
    \label{eqn:step1}
        \frac{\partial H(\lambda;d)}{\partial\lambda_i} = -\frac{\frac{\partial }{\partial \lambda_i}\int\mathcal{D}\phi\;p(\phi,\lambda|d)}{\int\mathcal{D}\phi\;p(\phi,\lambda|d)}
        =-\frac{\frac{\partial }{\partial \lambda_i}\int\mathcal{D}\phi\;p(\phi,\lambda|d)}{p(\lambda|d)}.
    \end{equation}
    We evaluate the numerator as
    $$
    \frac{\partial }{\partial \lambda_i}\int\mathcal{D}\phi\;p(\phi,\lambda|d) = \frac{\partial }{\partial \lambda_i}\int\mathcal{D}\phi\: \frac{1}{p(d)}p(d|\phi)\frac{\exp\left\{-\left[H(\phi;\lambda)+H(\lambda)\right]\right\}}{Z(\lambda)}.
    $$
    Now, if we pass $\partial/\partial\lambda_i$ through the integral, and apply the product rule and chain rule successively, we find
    \begin{align*}
        \frac{\partial }{\partial \lambda_i}\int\mathcal{D}\phi\;p(\phi,\lambda|d) = 
        \int\mathcal{D}\phi\:\frac{1}{p(d)}p(d|\phi)\Bigg\{
        &-\frac{\exp\left\{-\left[H(\phi;\lambda)+H(\lambda)\right]\right\}}{Z(\lambda)}\left[\frac{\partial H(\phi;\lambda)}{\partial\lambda_i}+\frac{\partial H(\lambda)}{\partial\lambda_i}\right]\\
        &- \frac{\exp\left\{-\left[H(\phi;\lambda)+H(\lambda)\right]\right\}}{Z(\lambda)}\frac{1}{Z(\lambda)}\frac{\partial Z(\lambda)}{\partial\lambda_i}
        \Bigg\}.
    \end{align*}
    Observe that the ratio $Z(\lambda)^{-1}\exp\left\{-\left[H(\phi;\lambda)+H(\lambda)\right]\right\}$, which appears in both terms, when multiplied by $p(d)^{-1}p(d|\phi)$ results in joint $p(\phi,\lambda|d)$.
    So, dividing the joint by the denominator of eq.~(\ref{eqn:step1}) the result is $p(\phi|d,\lambda)$ by Bayes's rule.
    Hence, we have shown 
    \begin{align*}
    \frac{\partial H(\lambda ;d)}{\partial \lambda_i} &=  \mathbb{E}\left[\frac{\partial H(\phi;\lambda)}{\partial \lambda_i} + \frac{\partial H(\lambda)}{\partial \lambda_i} + \frac{1}{Z(\lambda)}\frac{\partial Z(\lambda)}{\partial \lambda_i} \middle | d,\lambda\right] \\
    &= \mathbb{E}\left[ \frac{\partial H(\phi;\lambda)}{\partial \lambda_i}\middle | d,\lambda\right] + \frac{\partial H(\lambda)}{\partial \lambda_i} + \frac{1}{Z(\lambda)}\frac{\partial Z(\lambda)}{\partial \lambda_i},
    \end{align*}
    where we pulled the terms which do not depend on $\phi$ out of the field expectations.
    The derivative of the partition function gives
    \begin{align}
        \frac{1}{Z(\lambda)}\frac{\partial Z(\lambda)}{\partial \lambda_i} &= -\frac{1}{Z(\lambda)} \int \mathcal{D}\phi\: \exp\{-H(\phi;\lambda)\}\frac{\partial H(\phi;\lambda)}{\partial\lambda_i} \nonumber \\
        &= -\int \mathcal{D}\phi\: p(\phi|\lambda)\frac{\partial H(\phi;\lambda)}{\partial\lambda_i} \nonumber \\
        &= -\mathbb{E}\left[\frac{\partial H(\phi;\lambda)}{\partial\lambda_i} \middle|\lambda\right], \label{eqn:partderiv}
    \end{align}
    which completes the proof.
\end{proof}
This expression removes the problematic dependency on the partition function since $H(\phi;\lambda)$ is simply the field energy scaled by the model trust, i.e., $H(\phi;\lambda)=\beta E(\phi;u)$.
Under free theory, the expectations are taken over Gaussian random fields, and can be computed analytically.
We derive a similar result for the Hessian by relating $\nabla^2H(\lambda;d)$ to a covariance.
To this end, we need to define the concepts of posterior and prior covariance of a vector operator.
\begin{definition}[Conditional field covariance]
    \label{def:cov}
    Let $G : \mathcal{X} \to \mathbb{R}^m$ be an $m$-dimensional vector operator.
    Let $p(\phi|\lambda)$ and $p(\phi | d,\lambda)$ denote the field prior and posterior probability measures with fixed $\lambda$, respectively.
    We define the prior covariance of $G$ over the field prior to be the $m\times m$ matrix
    $$
    \mathbb{C}[G(\phi) | \lambda] ] \coloneqq \mathbb{E}\left[(G(\phi) - \mathbb{E}[G(\phi) |\lambda])(G(\phi) - \mathbb{E}[G(\phi) |\lambda])^T\middle| \lambda\right].
    $$
    Similarly, the posterior covariance of $G$ over the field posterior is
    $$
    \mathbb{C}[G(\phi) |d, \lambda] \coloneqq \mathbb{E}\left[(G(\phi) - \mathbb{E}[G(\phi) |d,\lambda])(G(\phi) - \mathbb{E}[G(\phi) |d,\lambda])^T\middle| d,\lambda\right].
    $$
\end{definition}

Using these definitions, we derive an equation for the desired Hessian:

\begin{lemma}
\label{thm:hessian}
    The Hessian of the parameter posterior potential is:
    \begin{align}
        \nabla_{\lambda}^2H(\lambda;d) &= \mathbb{C}\left[\nabla_\lambda H(\phi;\lambda)\middle|\lambda\right] -\mathbb{C}\left[\nabla_\lambda H(\phi;\lambda)\middle|d,\lambda\right] \nonumber \\
        &- \mathbb{E}\left[\nabla_{\lambda}^2H(\phi;\lambda)\middle|\lambda\right] + \mathbb{E}\left[\nabla_{\lambda}^2H(\phi;\lambda) \middle|d,\lambda\right] \nonumber \\
        &+\nabla^2_\lambda H(\lambda). \label{eqn:hessian}
    \end{align}
\end{lemma}

\begin{proof}
    See Appendix~\ref{apdx:hessian}.
\end{proof}

The form of the Hessian is illuminating and it allows us to derive a theorem that is useful for showing uniqueness in the common case of field potentials that are linear in the parameters.
We also remark that Lemma~\ref{thm:hessian} has potential to be useful in numerical algorithms.
In~\cite{alberts2023physics}, an SGLD scheme is used to draw samples from the posterior for inverse problems using a noisy estimate of the expectations appearing in Lemma~\ref{eqn:unbiased_grad_lambda}.
A discussion on exploiting second-order gradient information as a preconditioning matrix to improve the speed and accuracy of SGLD algorithms is provided in~\cite{li2016preconditioned}.
Lemma~\ref{thm:hessian} gives an explicit expression for the Hessian using expectations, which can be approximated directly from the samples.
Clearly, there is potential to improve such an SGLD approach by exploiting this expression, but developing such an algorithm is beyond the scope of this work.

\subsection{Weakly well-posed inverse problems}

Before posing a theorem on the uniqueness of MAP estimates to inverse problems in the case of field potentials which are linear in the parameters, we need to first state the following definitions.

\begin{definition}[Data density]
The data density is the probability density of the process that generates the data according to the assumed model, i.e.,
$$
p(d) = \int \mathcal{D}\phi\:\int d\lambda\: p(d|\phi)p(\phi|\lambda)p(\lambda).
$$
\end{definition}

\begin{definition}[Informative data]
    \label{def:informative}
    We say that the data are informative about the parameters for a given inverse problem if, for all data $d$ in the support of the data density, and for all parameters $\lambda$ in the support of the prior, the difference between the prior and posterior covariance matrices of the gradient of the field potential is positive definite.
    That is, the data are informative if
    $$
    \mathbb{C}\left[\nabla_\lambda H(\phi;\lambda)\middle|\lambda\right] - \mathbb{C}\left[\nabla_\lambda H(\phi;\lambda)\middle|d,\lambda\right] \succ 0.
    $$
\end{definition}

To understand this definition consider two extreme cases.
First, take the case when the posterior $p(\phi|d,\lambda)$ is the same as the prior $p(\phi|\lambda)$.
The difference between the two covariance terms in the above definition is clearly zero and the data are not informative about the parameters.
Second, suppose that the prior potential $H(\phi;\lambda)$ does not depend on the parameters.
Then, both covariance terms vanish identically and the data are not informative about the parameters.

Next, we first derive a theorem for the conditions in which the data are informative about the parameters.
We make use of the following notation:
\begin{notation}
    Let $S_{\lambda} : \mathcal{X} \to \mathcal{X}$ be an integral operator with kernel $s(\cdot,\cdot;\lambda)$, i.e.,
    $$
    (S_{\lambda}\phi)(x) \coloneqq \int s(x,y;\lambda)\phi(y)dy, \quad \phi \in \mathcal{X}.
    $$
    By $\nabla_{\lambda}S_{\lambda}$, we refer to the operator given by
    $$
    (\nabla_{\lambda}S_{\lambda}\phi)(x) \coloneqq \int \nabla_{\lambda}s(x,y;\lambda) \phi(y)dy, \quad \phi \in \mathcal{X}.
    $$
\end{notation}
The covariances in Lemma~\ref{thm:hessian} can be explicitly calculated under free theory using operator calculus described in~\cite{leike2016operator}.
We provide a brief primer to the technique, along with the proof, in Appendix~\ref{apdx:covariances}.
\begin{lemma}
\label{lem:covariance}
    Assume we are working in the free theory regime, following Def.~\ref{def:freetheory}, so that $H(\phi;\lambda) = 1/2\phi^{\dagger}S^{-1}_{\lambda}\phi$.
    Then, the covariance of \(\nabla_{\lambda}H(\phi;\lambda)\), taken over the prior $\mathcal{N}(0,S_{\lambda})$ is
    \begin{align*}
    \mathbb{C}\left[\nabla_{\lambda}H(\phi;\lambda)|\lambda\right] &= \frac{1}{2} \mathrm{tr}\left( \nabla_{\lambda} S^{-1}_{\lambda} S_{\lambda} \nabla_{\lambda} S^{-1}_{\lambda} S_{\lambda} \right),
    \end{align*}
    and the covariance of \(\nabla_{\lambda}H(\phi;\lambda)\), taken over the posterior $\mathcal{N}(\Tilde{m}_{\lambda}, \Tilde{S}_{\lambda})$ is
    \begin{align*}
    \mathbb{C}\left[\nabla_{\lambda}H(\phi;\lambda)|d, \lambda\right] &= \Tilde{m}_{\lambda}^{\dagger} \nabla_{\lambda} S_{\lambda}^{-1} \Tilde{S}_{\lambda} \nabla_{\lambda} S_{\lambda}^{-1} \Tilde{m}_{\lambda} \\
    &\quad + \frac{1}{2} \mathrm{tr}\left( \nabla_{\lambda} S^{-1}_{\lambda} \Tilde{S}_{\lambda} \nabla_{\lambda} S^{-1}_{\lambda} \Tilde{S}_{\lambda} \right).
    \end{align*}
\end{lemma}
\begin{proof}
    See Appendix~\ref{apdx:covariances}.
\end{proof}
The expressions in Lemma~\ref{lem:covariance} provide explicit representations of the covariances which appear in the Hessian of the parameter posterior potential, see Lemma~\ref{thm:hessian}.
These forms can be studied to derive situations in which the data are informative about the parameters.
Observe that Def.~\ref{def:informative} is satisfied when
\begin{align}
    \label{eqn:informativecond}
    \frac{1}{2} \mathrm{tr}\left( \nabla_{\lambda} S^{-1}_{\lambda} S_{\lambda} \nabla_{\lambda} S^{-1}_{\lambda} S_{\lambda} \right) - \Tilde{m}_{\lambda}^{\dagger} \nabla_{\lambda} S_{\lambda}^{-1} \Tilde{S}_{\lambda} \nabla_{\lambda} S_{\lambda}^{-1} \Tilde{m} - \frac{1}{2} \mathrm{tr}\left( \nabla_{\lambda} S^{-1}_{\lambda} \Tilde{S}_{\lambda} \nabla_{\lambda} S^{-1}_{\lambda} \Tilde{S}_{\lambda} \right) \succ 0.
\end{align}
Both of the negative terms in this expression contain $\Tilde{S}_{\lambda}$, which is the posterior covariance operator.
We will exploit the fact that under certain conditions, in the limit of infinite observations the posterior covariance kernel of a Gaussian process, and hence the associated covariance operator, $\Tilde{S}_{\lambda}$ vanishes.
Since the remaining term is positive, the result will be a positive-definite quantity.

\begin{proposition}
    \label{thm:informative_data}
    Suppose we are working under free theory.
    If, in the limit of infinite observations, the posterior covariance kernel vanishes, i.e., $||\Tilde{s}^{1/2}_{\lambda}||_{L^2(\Omega)} \to 0$ as $n \to \infty$, then the data are informative about the parameters given sufficient data.
\end{proposition}

\begin{proof}
    In order to show that the data are informative about the parameters, we must show
    $$
    \mathbb{C}\left[\nabla_\lambda H(\phi;\lambda)\middle|\lambda\right] - \mathbb{C}\left[\nabla_\lambda H(\phi;\lambda)\middle|d,\lambda\right] \succ 0.
    $$
    Expressions for the above covariances under free theory are given in Lemma~\ref{lem:covariance}, so the condition which must be satisfied is eq.~(\ref{eqn:informativecond}).
    The first term is trivially positive-definite, and does not depend on $n$.
    Since $||\Tilde{s}^{1/2}_{\lambda}||_{L^2(\Omega)} \to 0$ as $n \to \infty$, the associated posterior covariance operator $\Tilde{S}_{\lambda}$ also vanishes.
    This implies that the second and third terms can be taken to be arbitrarily small, depending on $n$.
    Then, given sufficient observations, eq.~(\ref{eqn:informativecond}) is positive-definite, and the data are informative about the parameters.
\end{proof}

The requirement for the posterior covariance collapsing to $0$ in the limit of infinite data is not satisfied in general.
However, the work found in~\cite{teckentrup2020convergence} provides some conditions for which this holds.
The theorems proved in that work are designed for situations in which the GP prior contains hyperparameters, in our case $\lambda$, which are estimated as a part of the inverse problem.
The conditions on the hyperparameters are lenient, and convergence holds under many common training schemes, for example a maximum likelihood or MAP estimate.
In fact, posterior convergence is independent on the hyperparameter training scheme.
To this end, we introduce the following definitions related to reproducing kernel Hilbert spaces (RKHSs), which can be found in~\cite{kanagawa2018gaussian}.
\begin{definition}[Kernel function]
    A symmetric mapping $s : \Omega \times \Omega \to \mathbb{R}$ is said to be a kernel function if for all $n \in \mathbb{N}$, $(c_1,\dots,c_n) \subset \mathbb{R}$, and $(x_1,\dots,x_n) \subset \Omega$, we have
    $$
    \sum_{i=1}^n\sum_{j=1}^n c_ic_js(x_i,x_j) \geq 0.
    $$    
\end{definition}

\begin{definition}[Reproducing kernel Hilbert space]
    \label{def:rkhs}
    Let $s: \Omega \times \Omega \to \mathbb{R}$ be a kernel function.
    A Hilbert space $H_s$ on $\Omega$ with inner product $\langle\cdot,\cdot\rangle_s$ is said to be a reproducing kernel Hilbert space if the following properties hold
    \begin{enumerate}[(i)]
        \item For all fixed $y \in \Omega$, $s(\cdot,y) \in H_s$.
        \item For all fixed $y \in \Omega$ and for all $\phi \in H_s$, $\phi(y) = \langle\phi, s(\cdot,y)\rangle_s$.
    \end{enumerate}  
\end{definition}

The second property of Def.~\ref{def:rkhs} is known as the reproducing property, and the kernel that satisfies this condition is known as the reproducing kernel.
The RKHS $H_s$ is uniquely determined by $s$, and the reverse also holds.
That is, given any kernel function $s$, there exists a RKHS such that $s$ is the reproducing kernel.
By definition for a kernel $s$ with RKHS $H_s$, each $\phi \in H_s$ can be written as $\phi = \sum_{i=1}^{\infty}c_is(\cdot,x_i)$ for some $c_i \in \mathbb{R}$, $x_i \in \Omega$, $i=1,2,\dots$, and $\|\phi\|_{H_s} < \infty$.
Therefore, one can observe that the fields within the RKHS have the same properties of $s$, e.g. regularity.

Next, we also need to place conditions on how the data are collected for convergence to hold.
In short, the data must be collected with some uniformity and follow a space-filling design.
\begin{definition}
    Let $X_n = (x_1,\dots,x_n) \subset \Omega$ denote the measurement locations.
    The \emph{fill distance} is given by
    $$
    h_{X_n} \coloneqq \sup_{x\in\Omega} \inf_{x_i\in X_n} \|x-x_i\|,
    $$
    the \emph{separation radius} is defined as
    $$
    r_{X_n} \coloneqq \frac{1}{2}\min_{i\neq j} \|x_i-x_j\|,
    $$
    and the \emph{mesh ration} is calculated as
    $$
    \rho_{X_n}\coloneqq \frac{h_{X_n}}{r_{X_n}}.
    $$
\end{definition}
The fill distance measures the maximum distance any $x\in\Omega$ can be from a data point $x_i \in X_n$, and the separation radius measures half the distance between distinct measurement locations.
Observe that both $h_{X_n}$ and $r_{X_n}$ go to $0$ as the number of data points go to infinity under a space-filling scheme, e.g., a uniform grid.
We will assume the measurements are taken uniformly so that $\rho_{X_n}$ is constant with $n$ to simplify calculations.
We use the following:
\begin{theorem}[Theorem 3.5~\cite{teckentrup2020convergence}]
    \label{thm:meanconv}
    Take $(\hat{\lambda}_i)_{i=1}^{\infty} \subseteq \Lambda$, $\Lambda \subseteq \mathrm{dom}(\lambda)$, compact, to be a sequence of estimates for $\lambda$.
    Assume the following conditions hold:
    \begin{enumerate}[(i)]
        \item $\Omega$ is compact with Lipschitz boundary and satisfies the interior cone condition.
        \item The RKHS associated with $s_{\lambda}$ is isomorphic to the Sobolev space $W^{\tau(\lambda)}(\Omega)$ for some $\tau(\lambda) \in \mathbb{N}$.
        \item The ground truth field $\phi^* \in W^{\bar{\tau}}(\Omega)$ for a $\bar{\tau} = \alpha + \eta$, where $\alpha \in \mathbb{N}$, $\alpha > d_{\Omega} / 2$, and $0 \leq \eta < 1$.
        \item The prior mean $m_{\lambda} \in W^{\bar{\tau}}(\Omega)$ for each $\lambda \in \Lambda$.
        \item There exists an $N^* \in \mathbb{N}$ such that the quantities $\tau^- \coloneqq \inf_{n\geq N^*} \tau(\hat{\lambda}_n)$ and $\tau^+ \coloneqq \sup_{n\geq N^*}\tau(\hat{\lambda}_n)$ satisfy $\Tilde{\tau} = \alpha' + \eta'$ for $\alpha' \in \mathbb{N}$, $\alpha' > d/2$ and $0 \leq \eta ' <1$.
    \end{enumerate}
    Then there exists a constant $C$, independent of $\phi^*$, $m_{\lambda}$, and $n$ such that for each $p \leq \bar{\tau}$, we have
    \begin{equation}
        \label{eqn:meanbound}
        \|\phi^* - \Tilde{m}_{\hat{\lambda}_n}\|_{W^p(\Omega)} \leq Ch_{X_n}^{\min\{\bar{\tau}, \tau^--p\}}\rho_{X_n}^{\max\{\tau^+-\bar{\tau},0\}}\left(\|\phi^*\|_{W^{\bar{\tau}}(\Omega)} + \sup_{n\geq N^*} \|m_{\hat{\lambda}_n}\|_{W^{\bar{\tau}}(\Omega)}\right),
    \end{equation}
    for all $n\geq N^*$ and $h_{X_n} \leq h_0$.
\end{theorem}
\begin{theorem}[Theorem 3.8~\cite{teckentrup2020convergence}]
    \label{thm:covconv}
    Let the assumptions of Theorem~\ref{thm:meanconv} hold.
    Then there exists a constant $C$, independent of $n$, such that
    \begin{equation}
        \label{eqn:covbound}
        \|\Tilde{s}_{\hat{\lambda}_n}^{1/2}\|_{L^2(\Omega)} \leq Ch_{X_n}^{\min\{\bar{\tau},\tau^-\} - d/2 - \varepsilon}\rho_{X_n}^{\max\{\tau^+-\bar{\tau},0\}},
    \end{equation}
    for each $n \geq N^*$, $h_{X_n} \leq h_0$, and $\varepsilon > 0$.
\end{theorem}

We now briefly discuss some of the conditions of Theorems~\ref{thm:meanconv} and~\ref{thm:covconv} under the IFT framework.
The first assumption is related to the regularity of the domain, and is trivially satisfied in most physical problems.
Likewise, assumption (iv) trivially holds under our free theory assumption, which enforces that $m_\lambda \equiv 0$.

Assumption (iii) is an assumption on the regularity of the ground truth being emulated.
Namely that $\phi^*$ must belong to a Sobolev space with a high-enough order.
Under our formulation of IFT, we are concerned with learning the solutions to PDEs, so this assumption is fairly easy to satisfy.
For example, the solutions to the heat equation and wave equation often belong to $W^2(\Omega)$~\cite{brezis2011functional}, which satisfies this condition up to $d_{\Omega} = 4$.
Also note that the condition holds for any field which is sufficiently regular and not just a solution to the assumed model.
This allows the results to hold even under the situation of model-form error, a primary interest of this work.

Next, assumption (v) places a constraint on how the hyperparameters $\lambda$ are learned.
One may view the two quantities $\tau^-$ and $\tau^+$ as equivalent to $\liminf \tau(\hat{\lambda}_n)$ and $\limsup \tau(\hat{\lambda}_n)$, respectively.
In this way, we observe that the assumption is a statement that the RKHS of the prior covariance $s_{\hat{\lambda}_n}$ must satisfy a certain smoothness as $\hat{\lambda}_n$ adjusts in order for the conclusion to hold.
This assumption follows immediately if $\lambda$ is fixed.

Finally, assumption (ii), which states that the RKHS of the prior covariance must be isomorphic to a Sobolev space, is the most difficult to verify.
It is known that this holds for classes of the Mat{\'e}rn and separable Mat{\'e}rn kernel families, including the exponential and Gaussian kernels~\cite{kanagawa2018gaussian}.
So, if the IFT prior is taken to be a Gaussian random field with such a covariance kernel (without the physics), then this assumption holds.

Of particular interest to our applications within IFT is the theory developed in~\cite{fasshauer2013reproducing}.
Here, there is a discussion on the RKHSs defined by Green's functions of PDEs.
Specifically, under relatively mild conditions for a differential operator $L$ and the boundary conditions, if the Green's function of $L$ is an even function, then the RKHS space of said Green's function is isomorphic to a Sobolev space (or even equivalent).
Under our framework of IFT, the prior is often constructed using a field energy principle for a PDE, and for certain operators the prior becomes a Gaussian process where the covariance kernel is a common Green's function scaled by $\beta^{-1}$, see for instance~\cite{alberts2023physics,ensslin2019information}.
Other approaches a priori take the covariance kernel to be a Green's function~\cite{albert2019gaussian, ranftl2022connection}.
Therefore if the specific Green's function in the application at hand induces a RKHSs which satisfies (ii), then we are free to apply Proposition~\ref{thm:informative_data}.

Note that the condition for the RKHS of the prior covariance kernel being isomorphic to a Sobolev space is not a strict requirement for the data to be informative about the parameters.
It may not be necessary in general for $\Tilde{S}_{\lambda}$ to vanish.
The only requirement is that the prior covariance of $\nabla_{\lambda}H(\phi;\lambda)$ is larger than the posterior covariance.
Intuitively, we would expect this to be true, as having access to more data should make us more sure about what we are inferring under a Bayesian framework, which should reduce the variance.
Determining if Def.~\ref{def:informative} holds can also be checked by numerically evaluating the expressions provided in eq.~(\ref{eqn:informativecond}).

We can now state the main theorem of the paper:
\begin{theorem}
\label{thm:theorem}
Suppose that the field potential is linear in the parameters $\lambda$, the parameter prior is constant, and the data are informative about the parameters.
If the gradient of the parameter posterior potential has a root, then the inverse problem is weakly well-posed under IFT.
\end{theorem}
\begin{proof}
Since the field potential is linear in the parameters, it can be decomposed into:
$$
H(\phi;\lambda) = \sum_{i=1}^m\lambda_i F_i(\phi),
$$
where $F_i \vcentcolon \mathcal{X} \to \mathbb{R}$ are suitable field functionals.
Then, the last three terms of Lemma~\ref{thm:hessian} vanish.
Since the data are informative about the parameters, the difference between the remaining two terms is a positive definite matrix.
Therefore, the parameter posterior potential $H(\lambda;d)$ is strictly convex.
Let $\lambda^*$ be such that $\nabla_{\lambda}H(\lambda^*;d)=0$.
Then, Proposition~\ref{lem:lemma1}, applies: $\lambda^*$ is the unique global minimizer of $H(\lambda;d)$, and the inverse problem is weakly well-posed.
\end{proof}

It is important to note here that Theorem~\ref{thm:theorem} is written in such a way that it holds even outside of free theory.
Although, outside of free theory, the condition that the data is informative about the parameters is much more difficult to verify.
Similar results can be derived for other common choices of priors for the parameters $\lambda$.
For example, if the parameters are each given a standard normal prior, then the Hessian $\nabla^2_{\lambda}H(\lambda)$ which appears in Lemma~\ref{thm:hessian} is simply an identity matrix, so the potential $H(\lambda;d)$ remains strictly convex if all other conditions remain the same.

We remark here that Theorem ~\ref{thm:theorem} covers many common cases within scientific and engineering applications.
What matters is the linearity of the parameters $\lambda$, and not the linearity of the field operators $F_i$, which are permitted to appear in a nonlinear way.
Oftentimes, the parameters of the physics appear in a linear manner, or a suitable change of variables may be made, where the new variables are linear, regardless of the linearity of the field functionals.
Usually, the functionals describe the field energy.
If the field energy is the variational form of the PDE, then parameters in the field operators retain their linearity coming from the PDE.
There are many examples found in engineering applications where the energy is linear in $\lambda$, including the heat equation, the wave equation, the equations of linear elasticity, the Allen-Cahn equation, among others.
Furthermore, this can apply to situations in which the field functionals contain additional fields that we infer as part of the inverse problem.
These fields (for example a spatial varying thermal conductivity) can be parameterized with a set of basis functions.
Alternatively, one could treat these fields as a finite-dimensional distribution of a Gaussian process.
To approximate this field, we simply infer the parameters of the basis functions, which appear in a linear way.

Caution must be taken in cases where the parameters do not appear in linear way.
For example, if the energy which goes into the prior is the integrated squared residual of a PDE, then there is a risk that the parameters multiply each other (if there are multiple parameters).
Furthermore, if we would like to infer the inverse temperature $\beta$, then simultaneously inferring any other parameters which go into the physics automatically makes the problem nonlinear, and Theorem~\ref{thm:theorem} does not apply.

\section{A free theory example}
\label{sec:example}
To demonstrate how the results of Sec.~\ref{sec:inverseprobs} above can be applied in theoretical investigations, we proceed with a free theory example of detecting model-form error under the Poisson equation.
Suppose we have a process which generates a data vector, $d_n = R_n\phi^* + \gamma$, $d_n\in\mathbb{R}^n$, where $R_n$ represents a linear measurement operator, taken to be point measurements of the field, and $\gamma \sim \mathcal{N}(0,\sigma^2I_n)$ is a noise process.
We let $\phi^*$ denote the true field which generates the data.
We believe that the field satisfies the two-dimensional Poisson equation with Dirichlet boundaries, i.e. $\phi^*$ is the unique solution to
\begin{equation}
    \label{eqn:poisson}
    -\nabla^2 \phi + q = 0,
\end{equation}
where $q \vcentcolon \Omega \to \mathbb{R}$ is the source term, subject to the boundary condition $\phi = 0$ on $\partial\Omega$.
The minus sign here enforces the differential operator defining eq.~(\ref{eqn:poisson}) to be positive-definite.
Given the data, we would like verify if eq.~(\ref{eqn:poisson}) represents the real, underlying ground-truth physics.
Is this a weakly well-posed problem?

In this situation, we are detecting model-form error, which we approach through the lens of IFT, as discussed in Sec.~\ref{sec:modelerrordisc}.
To detect model-form error, we begin by placing a flat prior on the trust, and look to infer $\beta$.
Equation~\ref{eqn:poisson} is known to have a variational formulation, which can serve as the basis for a physics-informed prior.
According to Dirichlet's principle~\cite{brezis2011functional}, it is $E(\phi) = \frac{1}{2} (\nabla \phi)^{\dagger}\nabla \phi - q^\dagger \phi$.
The energy defines a physics-informed prior of the type eq.~(\ref{eqn:piftprior}), where the the prior potential will be $H(\phi;\beta) = \beta E(\phi)$.
Under our formulation of IFT, the problem of detecting model-form error is solved by identifying $\beta$ as an inverse problem.
We will explore the weakly well-posedness of this inverse problem under IFT by applying Theorem~\ref{thm:theorem}.

The first step is to show that we are indeed working under free theory.
Define the operator $L(x,x') = -\delta(x-x')\nabla^2_{x'}$.
Then, one can show that an equivalent expression for the energy is $E(\phi) = \frac{1}{2} \phi^{\dagger} L \phi - q^{\dagger} \phi$.
To see this, begin by applying integration by parts on the first term of $E(\phi)$:
\begin{align*}
    (\nabla \phi)^{\dagger}\nabla \phi &= \int_{\Omega} \nabla \phi \cdot \nabla \phi \:dx\\
    &=
    \int_{\partial \Omega} \phi \: \nabla \phi \cdot \mathbf{n} \:dx - \int_{\Omega} \phi \: \nabla ^2 \phi \: dx,
\end{align*}
and the integral along the boundary vanishes from the boundary condition.
The equivalency can be shown using properties of Dirac's delta on the remaining integral:
$$
- \int_{\Omega} \phi \: \nabla ^2 \phi \:dx= -\int_{\Omega \times \Omega} \phi(x) \delta(x-x')\nabla_{x'}^2\phi(x')\:dxdx' = \phi^{\dagger}S^{-1}\phi.
$$

By completing the square, we identify that the potential is quadratic, and the IFT prior for this problem is a Gaussian random field.
The potential is 
$$
H(\phi;\beta) = \frac{1}{2} (\phi - Gq)^{\dagger} \beta L (\phi - Gq),
$$
where $G = L^{-1}$ is the integral operator with kernel $g$ given by the Green's function of eq.~(\ref{eqn:poisson}).
In the two-dimensional case, the Green's function is known to be $g(x,x') = -2\pi \log |x-x'|_2$~\cite[Chapter 8]{keener2018principles}.
Since the prior is a Gaussian random field, we write $p(\phi | \beta) = \mathcal{N}(\phi | Gq, \beta^{-1} G)$.

Notice that $\beta$, which does not appear in the mean function, is the only parameter we infer in the inverse problem.
To remain in free theory, we perform a change of variables, and infer the field $\psi = \phi - \mathbb{E}[\phi] = \phi - Gq$, so that $p(\psi | \beta) = \mathcal{N}(\psi|0, \beta^{-1} G).$
Since the measurement operator is linear, and the noise is Gaussian, we are under free theory, and the posterior is also a Gaussian random field, $p(\psi | d_n, \beta) = \mathcal{N}(\psi| \Tilde{m}_{\beta}, \Tilde{S}_{\beta})$, following Definition~\ref{def:freetheory}.

Since $\beta$ is the only parameter we infer, the field prior potential is linear in the parameters.
Further, we have placed a constant prior on $\beta$.
Then Theorem~\ref{thm:theorem} tells us we need to check two conditions.
First, we need to show that given sufficient measurements, the data are informative about $\beta$.
Second, we need to demonstrate the existence of a $\beta^*$ such that the gradient of the parameter posterior potential found in Lemma~\ref{eqn:unbiased_grad_lambda} vanishes.
Taking a flat prior for $\beta$ yields $\frac{\partial}{\partial \beta} H(\beta)=0$.
So, to demonstrate the existence of this $\beta^*$, we must show that there exists such a $\beta^*$ with:
\begin{equation}
    \label{eqn:cond2}
    \mathbb{E}\left[\frac{\partial}{\partial \beta} H(\psi;\beta) \middle|\beta=\beta^*\right] = \mathbb{E}\left[\frac{\partial}{\partial \beta} H(\psi;\beta) \middle| d,\beta=\beta^*\right].
\end{equation}
The gradient is straightforward, $\frac{\partial}{\partial \beta} H(\psi;\beta) = E(\psi)$, so eq.~(\ref{eqn:cond2}) says that the optimal choice of $\beta$ makes the expected prior and posterior energies equal.

To this end, we first evaluate the expectations.
Recall the definition of $\mathbb{E}[E(\psi);\beta]$:
$$
\mathbb{E}[E(\psi);\beta] = \int \mathcal{D}\psi \: E(\psi) p(\psi | \beta),
$$
and since $p(\psi | \beta) = \mathcal{N}(\psi|0, \beta^{-1} G)$, we note that this is simply an expectation over a Gaussian random field.
The same is true for the expectation taken over the posterior, and we find for both we must evaluate an expectation of the form $\mathbb{E}_{\psi\sim\mathcal{N}(m , D)}\left[ \frac{1}{2} \psi^{\dagger} S \psi \right]$, where $m$ and $D$ can be appropriately interchanged with the prior and posterior values.
We have already calculated expectations of this form, see eq.~(\ref{eqn:priorexp}) and eq.~(\ref{eqn:postexp}), where we simply need to make slight modifications.
The expressions give us
$$
\mathbb{E}[E(\psi)|\beta] = \frac{1}{2} \mathrm{tr}\left(\beta^{-1} G L\right)
$$
and
$$
\mathbb{E}[E(\psi)|d, \beta] = \frac{1}{2}\Tilde{m}_{\beta}^{\dagger} L \Tilde{m}_{\beta} + \frac{1}{2} \mathrm{tr}\left(\beta^{-1} G \Tilde{S}_{\beta}\right).
$$
For the expectation taken over the prior notice that since $G$ is the inverse of $L$, we have $\mathrm{tr}(GL) = \mathrm{tr}(\mathcal{I})$, the trace of the identity operator $\mathcal{I}$ on $\mathcal{X}$.
In the entire infinite-dimensional function space, $\mathrm{tr}(\mathcal{I})$ is infinite, and the expression is meaningless.
However, in practice we typically work on a finite-dimensional subspace of $\mathcal{X}$, a process called \emph{renormalization}, which is standard practice in quantum field theory applications~\cite[Chapter 10]{weinberg2005quantum}.
By doing so $\mathrm{tr}(\mathcal{I})$ is finite and equal to the number of dimensions of the subspace.
Typically, in IFT this is done by moving to the Fourier space, and truncating the space at the highest frequency of interest, as extreme frequencies are physically impossible.
For this reason we will take $\mathrm{tr}(\mathcal{I})$ to be finite.

We see that to prove the existence of a $\beta^*$ which makes the expectations equal, we must show there is a $\beta$ for which the following condition holds:
\begin{equation*}
    0 = \beta^{-1}\mathrm{tr}(\mathcal{I}) - \Tilde{m}_{\beta}^{\dagger} L \Tilde{m}_{\beta} - \beta^{-1} \mathrm{tr}\left(G \Tilde{S}_{\beta}\right).
\end{equation*}
Solving this expression for $\beta^{-1}$, we find
\begin{equation}
    \label{eqn:condition}
    \beta^{-1} = \frac{\Tilde{m}_{\beta}^{\dagger} L \Tilde{m}_{\beta} }{\mathrm{tr}\left(\mathcal{I}\right) - \mathrm{tr}\left(G\Tilde{S}_{\beta}\right)}.
\end{equation}
To understand eq.~(\ref{eqn:condition}) we study the extreme case of infinite data.

Begin by noting that the covariance kernel of the prior is simply the Green's function of eq.~(\ref{eqn:poisson}), scaled by the inverse trust.
The RKHS imposed by said Green's function is isomorphic to the Sobolev space $W^1(\Omega)$, see~\cite{fasshauer2013reproducing}.
Then, Theorems~\ref{thm:meanconv} and~\ref{thm:covconv} apply, and as $n \to \infty$ the posterior covariance vanishes, i.e. $||\Tilde{s}^{1/2}_{\beta}||_{L^2(\Omega)} \to 0$.
Furthermore under these conditions, the posterior mean converges to the underlying field which generated the data, independent of the choice of $\beta$ (provided that $\beta$ is not zero or infinite), i.e. $||\psi^* - \Tilde{m}_{\beta}||_{L^2(\Omega)} \to 0$.

The covariance vanishing in the limit of infinite data implies that the data are informative about $\beta$ under Proposition~\ref{thm:informative_data}, so the first requirement is satisfied.
Letting $n \to \infty$, and undoing the change of variables eq.~(\ref{eqn:condition}) becomes
\begin{equation}
    \label{eqn:optimalbeta}
    \beta^{*^{-1}} = \frac{1}{\mathrm{tr}(\mathcal{I})}\left(\phi^* - Gq\right)^{\dagger} L \left(\phi^* - Gq\right).
\end{equation}

From here, we make a few observations about the optimal value of trust.
If $\phi^*$ is exactly the field which solves eq.~(\ref{eqn:poisson}) (including boundary conditions), then $\phi^* = Gq$.
So the right-hand-side of eq.~(\ref{eqn:optimalbeta}) becomes zero, and there is no finite $\beta^*$ which satisfies the relationship.
The $\beta^*$ which solves the inverse problem must be infinite.
This is consistent with our interpretation of physics-informed IFT.
If we have selected the correct physical model, then IFT tells us to believe the physics to be true, and we should find the underlying field by directly solving the PDE.

Under model-form error, $\phi^*$ is not the solution of eq.~(\ref{eqn:poisson}), and $\phi^* - Gq \neq 0$.
Since $L$ is positive-definite and bounded, $0 < \left(\phi^* - Gq\right)^{\dagger} L \left(\phi^* - Gq\right) < \infty$ thus $\beta^*$ is finite and positive.
Therefore, under the limit of infinite data, Theorem~\ref{thm:theorem} applies, and the model-form error detection problem is a weakly well-posed inverse problem.
This unique value of trust tells us to what extent we should believe the physics to be true.
As $\phi^*$ gets closer to satisfying the chosen physical model, the difference between $\phi^*$ and the solution of eq.~(\ref{eqn:poisson}) goes to $0$.
As a result, the optimal trust is larger when the physics become more correct.

Finally, we remark that the optimal value of $\beta$ depends on the dimension of the subspace of $\mathcal{X}$ chosen since $\mathrm{tr}(\mathcal{I})$ appears in eq.~(\ref{eqn:condition}).
As the dimension of this subspace grows, the optimal value of $\beta$ also grows.
The intuitive understanding of this fact is that we should expect less model-form error, hence a larger $\beta$, as the mesh for $\mathcal{X}$ becomes finer.
That is, we induce less model-form error as we improve our ability to capture the finer features of the field.

We conclude with the case of an uninformative, Jeffreys prior on $\beta$.
That is $p(\beta)\propto 1/\beta$ or, equivalently, $H(\beta) = \log\beta$.
Then, Theorem~\ref{thm:theorem} does not apply and we have to rely on Proposition~\ref{lem:lemma1}.
The second derivative of the parameter posterior potential is:
$$
\frac{\partial^2}{\partial\beta^2} H(\beta;d) = \mathbb{C}[E(\psi)|\beta] - \mathbb{C}[E(\psi)|d,\beta] - \frac{1}{\beta^2}.
$$
We see that now it is not sufficient for the data to be informative about $\beta$.
Instead, the condition is that the difference between the prior and posterior energy variance should be greater than $1/\beta^2$ -- not just positive.
Observe how this quantity changes with $\beta$.
As $\beta$ increases, which corresponds to the assumption that the model is more correct, the required lower bound gets smaller.
As $\beta$ decreases towards zero, the required lower bound goes to infinity.
The latter condition reflects the fact that under a Jeffreys prior, $\beta$ cannot be exactly zero.

Under a Jeffreys prior, in order to make $\frac{\partial}{\partial \beta} H(\beta ;d)$ vanish, the optimal $\beta$ must satisfy:
$$
0 = \mathbb{E}[E(\psi)|\beta=\beta^*] - \mathbb{E}[E(\psi)|d,\beta=\beta^*] - \frac{1}{\beta^*}.
$$
The expectations are exactly the same as with the flat prior, which yields the condition
\begin{equation*}
    \beta^{-1} = \frac{\Tilde{m}^{\dagger}_{\beta} L \Tilde{m}_{\beta}}{\mathrm{tr}(\mathcal{I}) - \mathrm{tr}\left(G\Tilde{S}_{\beta}\right) -2}.
\end{equation*}
Taking the limit of infinite data and reverting to the original variables, we see the optimal trust is given by
\begin{equation*}
    \label{eqn:optimal_jeff_beta}
        \beta^{*^{-1}} = \frac{1}{\mathrm{tr}(\mathcal{I}) - 2}\left(\phi^* - Gq\right)^{\dagger} L \left(\phi^* - Gq\right).
\end{equation*}
To study well-posedness, we study the convexity of $H(\beta;d)$.
Evaluating $\frac{\partial^2}{\partial\beta^2} H(\beta;d)$ at this $\beta^*$, according to Lemma~\ref{thm:hessian}, and using the expressions for the covariances derived in Lemma~\ref{lem:covariance}, we find
$$
\frac{\partial^2}{\partial\beta^2} H(\beta = \beta^*|d) = \frac{1}{2} \left(\beta^{*^{-1}}\right)^2\left(\mathrm{tr}(\mathcal{I})-1\right).
$$
Therefore, we find that a sufficient condition for both $\beta^* > 0$ and $\frac{\partial^2}{\partial\beta^2} H(\beta;d) > 0$ is $\mathrm{tr}\left(\mathcal{I}\right) > 2$, meaning that the subspace of $\mathcal{X}$ taken must have a more than 2 dimensions.
Practically this means that at the bare minimum, we must employ a linear model for the field.
In this case, the inverse problem of detecting model-form error is weakly well-posed.

\section{Discussion and Conclusions}
\label{sec:conclusions}
In this work, we studied the application of IFT to inverse problems.
Particularly, we derived conditions for which an inverse problem under IFT yields a unique MAP estimate.
Under free theory, the mathematics of Gaussian random fields can be used to characterize these conditions analytically, and we found an inverse problem becomes weakly well-posed under relatively mild assumptions.
While IFT was the main focus of this paper, a relationship to Gaussian process regression in general can be established to study how hyperparameters of Gaussian processes are trained.
However, many commonly used covariance kernels are not linear in the parameters, so Theorem~\ref{thm:theorem} does not apply, and Proposition~\ref{lem:lemma1} must be used.
Additionally, the posterior covariance kernel vanishing in the limit of infinite data was assumed, but this does not hold for a Gaussian process prior in general.
This condition can be relaxed if expressions appearing in Proposition~\ref{thm:informative_data} can be characterized.
The uniqueness of the MAP estimator to an inverse problem depends on both the gradient and the Hessian of the parameter posterior potential, which can be computed using expectations over the field.
Under Gaussian process regression, the expectations could perhaps be computed using the operator calculus for IFT~\cite{leike2016operator}.

The discussion was mainly limited to field potentials which are linear in the parameters so that the gradient and Hessian are easier to characterized.
While this is likely the case for field energies coming from a variational formulation of a PDE, this is not the case in general.
As mentioned above, the parameters in many classic Gaussian process covariance kernels appear in a nonlinear manner, such as the parameters of a Mat{\'e}rn kernel.
If the field energy is taken to be the integrated square residual of a PDE, then the parameters will also appear nonlinearly.
Studying the well-posedness of the parameters in the integrated square residual case is an important task, since it is most commonly used as an approximation to the field energy in practice.
In addition to this, inferring the trust along with other parameters automatically makes the problem nonlinear.
There are some situations where we would like to identify both, and finding the conditions under which the problem is weakly well-posed would be useful.
This problem is not weakly well-posed in general, e.g., for the heat equation both the trust parameter and the thermal conductivity will appear together.

In Sec.~\ref{sec:example}, an analytical example is provided with a focus on detecting model-form error under the Poisson equation.
We found that in the case of model-form error, the problem selecting the optimal trust becomes weakly well-posed in the limit of infinite data.
The expression for the optimal trust derived in this case reveals two important points.
If the physics are exactly correct, then the theory tells us to select infinite trust in the model.
Under model-form error, the optimal trust will be larger for models which are more correct.
This coincides with our current understanding of the physics-informed IFT approach.
This result was derived with a flat prior on the trust.
Finally, the example provided was constructed in the free theory case with both a flat prior and Jefferys prior selected for the trust parameter.
Insights with other priors or and with other common PDEs could prove useful.

\bibliographystyle{plain}
\bibliography{references}

\appendix

\section{Proof of lemma~\ref{thm:hessian}}
\label{apdx:hessian}
The proof is tedious but straightforward.
We need to find the Hessian matrix elements $\partial^2H(\lambda;d)/\partial\lambda_i\partial\lambda_j$.
From Lemma~\ref{eqn:unbiased_grad_lambda}, we have:
\begin{equation}
    \label{eqn:hessian_exp_0}
    \frac{\partial^2H(\lambda;d)}{\partial\lambda_i\partial\lambda_j} = 
    \frac{\partial}{\partial\lambda_i}\left\{
    \mathbb{E}\left[\frac{\partial H(\phi;\lambda)}{\partial\lambda_j}\middle|d,\lambda\right]
    -
    \mathbb{E}\left[\frac{\partial H(\phi;\lambda)}{\partial\lambda_j}\middle|\lambda\right]
    \right\}
    + \frac{\partial^2H(\lambda)}{\partial\lambda_i\partial\lambda_j}.
\end{equation}
So, the last term of eq.~(\ref{eqn:hessian}) is already in place.
We carry out the derivatives of the expectations one by one.

First we look at the derivative of the posterior expectation.
Expressing the expectation as a path integral, passing the derivative inside the integral, using the product rule of differentiation, and splitting the integral in two summands yields:
\begin{equation}
    \label{eqn:hessian_exp_1}
    \frac{\partial}{\partial\lambda_i} \mathbb{E}\left[\frac{\partial H(\phi;\lambda)}{\partial\lambda_j}\middle|d,\lambda\right]
    =
    \int\mathcal{D}\phi\:\frac{\partial p(\phi|d,\lambda)}{\partial\lambda_i}\frac{\partial H(\phi;\lambda)}{\partial\lambda_j}
    +
    \int\mathcal{D}\phi\:p(\phi|d,\lambda)\frac{\partial^2H(\phi;\lambda)}{\partial\lambda_i\partial\lambda_j}.
\end{equation}
The second path integral is the expectation $\mathbb{E}\left[\partial^2H(\phi;\lambda)/\partial\lambda_i\partial\lambda_j\middle|d,\lambda\right]$ which appears as is in eq.~(\ref{eqn:hessian}).

To resolve the first path integral of eq.~(\ref{eqn:hessian_exp_1}), we start with the derivative of the field posterior conditional on $\lambda$.
Using Bayes's rule, followed by product rule of differentiation, the chain rule, and then again Bayes's rule, we have:
\begin{align}
    \frac{\partial p(\phi|d,\lambda)}{\partial\lambda_i} &= \frac{\partial }{\partial\lambda_i}\left[\frac{p(\phi,\lambda|d)}{p(\lambda)}\right] \nonumber \\
    &= \frac{\partial p(\phi,\lambda|d)}{\partial\lambda_i}\frac{1}{p(\lambda)} - \frac{p(\phi,\lambda|d)}{p(\lambda)^2}\frac{\partial p(\lambda)}{\partial\lambda_i} \nonumber \\
    &= \frac{\partial p(\phi,\lambda|d)}{\partial\lambda_i}\frac{1}{p(\lambda)} - p(\phi|d,\lambda)\frac{\partial p(\lambda)}{\partial\lambda_i}\frac{1}{p(\lambda)} \label{eqn:hessian_exp_1_a}.
\end{align}
Recalling that $p(\lambda)=\exp\{-H(\lambda)\}$, we reduce the second term:
\begin{equation*}
    \frac{\partial p(\lambda)}{\partial\lambda_i}\frac{1}{p(\lambda)} = -\frac{\partial H(\lambda)}{\partial\lambda_i}\frac{\exp\{-H(\lambda)\}}{p(\lambda)}=-\frac{\partial H(\lambda)}{\partial\lambda_i}.
\end{equation*}
For eq.~(\ref{eqn:hessian_exp_1}), we need the path integral of this quantity multiplied by $p(\phi|d,\lambda)\frac{\partial H(\phi;\lambda)}{\partial\lambda_j}$.
Recalling the definition of posterior expectation and using the fact that the second factor is a constant with respect to the field, we get:
\begin{equation*}
    \int\mathcal{D}\phi\:p(\phi|d,\lambda)\frac{\partial p(\lambda)}{\partial\lambda_i}\frac{1}{p(\lambda)}\frac{\partial H(\phi;\lambda)}{\partial\lambda_j}
    =-\mathbb{E}\left[\frac{\partial H(\phi;\lambda)}{\partial\lambda_j}\middle| d,\lambda\right]\frac{\partial H(\lambda)}{\partial\lambda_i}.
\end{equation*}

For the first term in eq.~(\ref{eqn:hessian_exp_1_a}), use the product rule of differentiation, followed by the chain rule, take out the common factors, and apply Bayes's rule again to obtain:
\begin{align*}
    \frac{1}{p(\lambda)} \frac{\partial p(\phi,\lambda|d)}{\partial\lambda_i} &= \frac{1}{p(\lambda)}
    \frac{\partial }{\partial\lambda_i}\frac{\exp\left\{-\left[H(\phi;\lambda) + H(\lambda)\right]\right\}}{Z(\lambda)}\\
    &= 
     \frac{1}{p(\lambda)}\left\{-\frac{\exp\left\{-H(\phi;\lambda) - H(\lambda)\right\}}{Z(\lambda)}\left[\frac{\partial H(\phi;\lambda)}{\partial\lambda_i}+\frac{\partial H(\lambda)}{\partial\lambda_i}\right] 
     -\frac{\exp\left\{-H(\phi;\lambda) - H(\lambda)\right\}}{\left(Z(\lambda)\right)^2}\frac{\partial Z(\lambda)}{\partial \lambda_i}\right\}\\
    &=  
     \frac{1}{p(\lambda)}\left\{-p(\phi|d,\lambda)\left[\frac{\partial H(\phi;\lambda)}{\partial\lambda_i}+\frac{\partial H(\lambda)}{\partial\lambda_i}\right] 
     -p(\phi|d,\lambda)\frac{1}{Z(\lambda)}\frac{\partial Z(\lambda)}{\partial \lambda_i}\right\}\\
     &=
     -\frac{p(\phi|d,\lambda)}{p(\lambda)}\left\{\frac{\partial H(\phi;\lambda)}{\partial\lambda_i}+\frac{\partial H(\lambda)}{\partial\lambda_i} 
     +\frac{1}{Z(\lambda)}\frac{\partial Z(\lambda)}{\partial \lambda_i}\right\}\\
     &= 
     -p(\phi|d,\lambda)\left\{\frac{\partial H(\phi;\lambda)}{\partial\lambda_i}+\frac{\partial H(\lambda)}{\partial\lambda_i} 
    +\frac{1}{Z(\lambda)}\frac{\partial Z(\lambda)}{\partial \lambda_i}\right\}.
\end{align*}
For eq.~(\ref{eqn:hessian_exp_1}), we need the path integral of this quantity multiplied by $\partial H(\phi|\lambda)/\partial\lambda_j$.
Using the definition of the posterior expectation, taking constants with respect to the field out of the expectation, and employing eq.~(\ref{eqn:partderiv}) yields:
\begin{align*}
    \int\mathcal{D}\phi\:\frac{1}{p(\lambda)} \frac{\partial p(\phi,\lambda|d)}{\partial\lambda_i}\frac{\partial H(\phi;\lambda)}{\partial\lambda_j} =
    &-\mathbb{E}\left[\frac{\partial H(\phi;\lambda)}{\partial\lambda_i}\frac{\partial H(\phi;\lambda)}{\partial\lambda_j}\middle|d,\lambda\right] \\
    &-\mathbb{E}\left[\frac{\partial H(\phi;\lambda)}{\partial\lambda_j}\middle| d,\lambda\right]\frac{\partial H(\lambda)}{\partial\lambda_i} \\
    &+\mathbb{E}\left[\frac{\partial H(\phi;\lambda)}{\partial\lambda_j}\middle| d,\lambda\right]\mathbb{E}\left[\frac{\partial H(\phi;\lambda)}{\partial\lambda_i}\middle|d,\lambda\right].
\end{align*}
Notice that the first and the third lines on the right-hand side of the equation above give minus the posterior covariance of two quantities:
\begin{equation*}
    \int\mathcal{D}\phi\:\frac{1}{p(\lambda)} \frac{\partial p(\phi,\lambda|d)}{\partial\lambda_i}\frac{\partial H(\phi;\lambda)}{\partial\lambda_j} =
    -\mathbb{C}\left[\frac{\partial H(\phi;\lambda)}{\partial\lambda_i},\frac{\partial H(\phi;\lambda)}{\partial\lambda_j}\middle|d,\lambda\right]
    -\mathbb{E}\left[\frac{\partial H(\phi;\lambda)}{\partial\lambda_j}\middle| d,\lambda\right]\frac{\partial H(\lambda)}{\partial\lambda_i}.
\end{equation*}
Plugging the above results into eq.~(\ref{eqn:hessian_exp_1}) and canceling the two opposite terms that arise results in:
\begin{equation}
    \label{eqn:hessian_exp_1_b}
    \frac{\partial}{\partial\lambda_i} \mathbb{E}\left[\frac{\partial H(\phi;\lambda)}{\partial\lambda_j}\middle|d,\lambda\right] 
    = -\mathbb{C}\left[\frac{\partial H(\phi;\lambda)}{\partial\lambda_i},\frac{\partial H(\phi;\lambda)}{\partial\lambda_j}\middle|d,\lambda\right]
    + \mathbb{E}\left[\frac{\partial^2H(\phi;\lambda)}{\partial\lambda_i\partial\lambda_j}\middle|d,\lambda\right].
\end{equation}
Following similar steps for the prior expectation in eq.~(\ref{eqn:hessian_exp_0}) yields:
\begin{equation}
    \label{eqn:hessian_exp_1_c}
    \frac{\partial}{\partial\lambda_i} \mathbb{E}\left[\frac{\partial H(\phi;\lambda)}{\partial\lambda_j}\middle|\lambda\right] 
    = -\mathbb{C}\left[\frac{\partial H(\phi;\lambda)}{\partial\lambda_i},\frac{\partial H(\phi;\lambda)}{\partial\lambda_j}\middle|\lambda\right]
    + \mathbb{E}\left[\frac{\partial^2H(\phi;\lambda)}{\partial\lambda_i\partial\lambda_j}\middle|\lambda\right].
\end{equation}
Plugging in the right-hand sides of eq.~(\ref{eqn:hessian_exp_1_b}) and eq.~(\ref{eqn:hessian_exp_1_c}) into eq.~(\ref{eqn:hessian_exp_0}) completes the proof.

\section{Operator calculus and proof of Lemma~\ref{lem:covariance}}
\label{apdx:covariances}
To prove Lemma~\ref{lem:covariance}, we rely heavily on the operator calculus method in~\cite{leike2016operator}.
Before proving the result, we provide a brief background on the method.

\subsection{Primer on operator calculus for free theory}

The proof involves calculating various expectations of functionals over Gaussian random fields.
Using our established notation under IFT, we may represent such an expectation as the path integral
$$
\mathbb{E}[F(\phi) | \mathcal{N}(m,S)] = \int \mathcal{D}\phi\: F(\phi) p(\phi),
$$
where $p(\phi)$ is the density of $\mathcal{N}(m,S)$, i.e.,
$$
p(\phi) = \frac{1}{\sqrt{2\pi \det S}}\exp\left\{-\frac{1}{2}(\phi - m)^{\dagger} S^{-1}(\phi - m)\right\}.
$$
The operator calculus method provides tools for evaluating these expectations in a systematic way.
The method works by turning expectations over Gaussian random fields into noncommutative algebraic operations.

Suppose we wish to evaluate the expectation $\mathbb{E}[F(\phi) | \mathcal{N}(m,S)]$, where $F : \mathcal{X}\to \mathbb{R}$ is an analytical functional.
We define the so-called field-operator $\Phi : \mathcal{X} \to \mathcal{X}$ to be
\begin{equation}
    \label{eqn:fieldop}
    (\Phi\phi)(x) \coloneqq \left(\left[m + \int dy\: s(x,y) \frac{\delta}{\delta m(y)}\right]\phi\right)(x),
\end{equation}
where $s(x,y)$ is the kernel of $S$, and $\delta/\delta m$ denotes `action' of computing the gradient of a functional $F$ at $m$.
Physicists may recognize this as the variation of $F$ in the direction of the Dirac delta. 
To be precise, we state the definition of the derivative of a functional:
\begin{definition}[G\^ateaux derivative of a real-valued functional]
    Let $F : \mathcal{X} \to \mathbb{R}$ be a functional of $\phi$.
    The G\^ateaux derivative of $F$ at a point $\phi \in \mathcal{X}$ in the direction of the function $\psi \in \mathcal{X}$ is given by
    $$
    \frac{\delta F(\phi)}{\delta \psi} \coloneqq \lim_{\varepsilon \to 0} \frac{F(\phi + \varepsilon \psi) - F(\phi)}{\varepsilon}.
    $$
\end{definition}
The gradient of a functional is defined by relating the G\^ateaux derivative to an inner product in $L^2(\Omega)$.
By definition, the G\^ateaux derivative is continuous and linear (assuming that $F$ is also Fr\^echt differentiable), hence it is bounded, and the Riesz representation theorem applies~\cite{kreyszig1991introductory}[Theorem 3.8-1].
Therefore, the G\^ateaux derivative can be represented as an $L^2(\Omega)$ inner product.
That is, assuming that $F$ is differentiable at $m$, there exists a unique function in $L^2(\Omega)$, which is said to be the gradient of $F$, denoted by $\delta F(m) /\delta m(y)$, such that
$$
\frac{\delta F(\phi)}{\delta \psi} = \int_{\Omega} \frac{\delta F(\phi)}{m(y)} \psi d\Omega, \quad \forall \psi \in L^2(\Omega).
$$

Returning to eq.~(\ref{eqn:fieldop}), in short-hand notation we let $b_x = m$, which is termed the \emph{creation} operator, and we let $c_x = \int dy\: s(x,y)\frac{\delta}{\delta m(y)}$, which is called the \emph{annihilation operator}.
Then, to compute the expectation $\mathbb{E}[F(\phi) | \mathcal{N}(m,S)]$, we simply evaluate $F(\Phi)1$.
That is, we define a new operator $F(\Phi) : \mathcal{X} \to \mathcal{X}$ and let it act on the function $1$.
Performing the calculation $F(\Phi)1$ can be done by separating $\Phi$ into the creation and annihilation operators.
It is crucial to note here that the creation and annihilation operators do not commute.
However, we can evaluate products of creation and annihilation operators with the commutator $[c_x,b_y] \coloneqq c_xb_y - b_yc_x$.
Observe that $c_x$ and $b_y$ are defined in such a way that $[c_x,b_y] = S$ since we have
\begin{align*}
    c_xb_yF(m) &= \int s(x,x') \frac{\delta}{\delta m(x')}\{m(y) F(m)\}dx' \\
    &= \int s(x,x') \left\{\frac{\delta m(y)}{\delta m(x')} F(m) + m(y) \frac{\delta F(m)}{\delta m(x')}\right\} dx' \\
    &= \int s(x,x') \left\{F(m) + m(y) \frac{\delta F(m)}{\delta m(x')}\right\}dx' \\ 
    &= SF(m) + m(y) \int s(x,x') \frac{\delta F(m)}{\delta m(x')} dx' \\
    &= SF(m) + m(y) \left(\int dx' s(x,x')\frac{\delta}{\delta m(x')}\right)F(m) \\
    &= SF(m) + b_yc_xF(m) \\
    &\implies [c_x, b_y] = S.
\end{align*}
Then, any calculation involving $\Phi$ can be performed by moving the annihilation operators to the right-hand side of the equation using properties of the commutator, as we have $c_x1 = 0$.
For example, $c_xb_y1 = [c_x,b_y]1 + b_yc_x1 = S$.

\subsection{Proof of Lemma~\ref{lem:covariance}}

We expand the covariances using expectations
\begin{align}
    \mathbb{C}\left[\nabla_\lambda H(\phi;\lambda)\middle|\lambda\right] &= \mathbb{E}\left[\left[\nabla_\lambda H(\phi;\lambda)\right]^{T}\left[\nabla_\lambda H(\phi;\lambda)\right]\middle|\lambda\right] \nonumber \\
    &- \mathbb{E}\left[\nabla_\lambda H(\phi;\lambda)\middle|\lambda\right]^{T}\left[\mathbb{E}\left[\nabla_\lambda H(\phi;\lambda)\middle|\lambda\right]\right], \label{eqn:priorcov}
\end{align}
and
\begin{align}
    \mathbb{C}\left[\nabla_\lambda H(\phi;\lambda)\middle|d,\lambda\right] &= \mathbb{E}\left[\left[\nabla_\lambda H(\phi;\lambda)\right]^T\left[\nabla_\lambda H(\phi;\lambda)\right]\middle|d,\lambda\right] \nonumber \\
    &- \mathbb{E}\left[\nabla_\lambda H(\phi;\lambda)\middle|d,\lambda\right]^T\left[\mathbb{E}\left[\nabla_\lambda H(\phi;\lambda)\middle|d,\lambda\right]\right]. \label{eqn:postcov}
\end{align}
We evaluate each expectation one by one.
First, to simplify the notation moving forward, we drop the subscript $\lambda$, understanding that the derivatives are taken with respect to $\lambda$.
Further, we adopt the Einstein summation convention to simplify writing the integrals which will appear.
So, for example, we write
$$
\nabla S^{-1}_{xy} \coloneqq \nabla_{\lambda} S^{-1}_{\lambda},
$$
and under the Einstein summation convention, for functions $\phi, \psi \in \mathcal{X}$, we write
\begin{equation*}
    \psi^{\dagger}\phi = \phi_x \psi_x \coloneqq \int \phi(x) \psi(x) dx,
\end{equation*}
and
\begin{equation*}
    (S\phi)_x = S_{xy} \phi_y \coloneqq \int s(x,y) \phi(y) dy.
\end{equation*}
Furthermore, note that the trace of an operator is written as
\begin{equation}
    \label{eqn:trace}
    S_{xx} \coloneqq \int_{\Omega} s(x,x)dx = \mathrm{tr}(S),
\end{equation}
or for a combination of operators, we have $A_{xy}S_{yx} = \mathrm{tr}(AS)$.

We make extensive use of the operator calculus to evaluate the expectations in the equations~(\ref{eqn:priorcov}) and~(\ref{eqn:postcov}).
Since we are in the free theory case, we know that (i) the prior is Gaussian, following $p(\phi|\lambda) = \mathcal{N}(\phi|0,S)$, (ii) the posterior is also Gaussian, with $p(\phi | d, \lambda) = \mathcal{N}(\phi|\Tilde{m}, \Tilde{S})$, with both $\Tilde{m}$ and $\Tilde{S}$ given in Def.~\ref{def:freetheory}, and finally (iii) the potential will be of quadratic form, i.e. $H(\phi;\lambda) = \frac{1}{2}\phi^{\dagger}S^{-1}\phi$.
From here we see the gradient of the potential is simply $\nabla H(\phi ; \lambda) = \frac{1}{2}\phi^{\dagger}\nabla S^{-1}\phi$.

Everything is in place to evaluate the expectations, and we proceed with the expectation of $\nabla H(\phi;\lambda)$ first.
We need to evaluate the expectation of this expression over the prior and posterior, which are both Gaussian.
So, we will evaluate $\mathbb{E} \left[ \frac{1}{2}\phi^{\dagger}\nabla S^{-1}\phi \middle | \mathcal{N}(m, D) \right]$ for a general $m$ and $D$.
Then we will input the appropriate prior and posterior mean and covariance at the end of the calculation.
We find the following result.
\begin{lemma}
    \label{lem:C1}
    The expectation of $\nabla H(\phi;\lambda)$, taken over a Gaussian random field with mean $m$ and covariance $D$, is given by
    $$
    \mathbb{E} \left[ \frac{1}{2}\phi^{\dagger}\nabla S^{-1}\phi \middle | \mathcal{N}(m, D) \right] = \frac{1}{2}m^{\dagger}\nabla S^{-1}m + \frac{1}{2}\mathrm{tr}\left(\nabla S^{-1}D\right).
    $$
\end{lemma}
\begin{proof}
    The analytic functional we are taking the expectation of is $F(\phi) = \frac{1}{2}\phi^{\dagger}\nabla S^{-1}\phi$, and to evaluate the expectation we must work out $F(\Phi)1$.
    Expanding $F(\phi)$ as an inner product, pulling the integrals out of the expectation, expressing $\Phi$ in terms of the creation and annihilation operators, and inserting them in place of $\phi$ and letting it act on $1$, we obtain
    \begin{align*}
        F(\Phi)1 = \frac{1}{2} \Phi_x\nabla S^{-1}_{xy}\Phi_y 1 &=
        \frac{1}{2} (c_x + b_x)\nabla S^{-1}_{xy}(c_y + b_y)1 \\
        &= 
        \frac{1}{2}(c_x + b_x)\nabla S^{-1}_{xy}b_y 1\\
        &=
        \frac{1}{2}c_x\nabla S^{-1}_{xy}b_y 1+ \frac{1}{2}b_x\nabla S^{-1}_{xy}b_y 1\\
        &=
        \frac{1}{2}\nabla S^{-1}_{xy}c_xb_y 1+ \frac{1}{2}b_x\nabla S^{-1}_{xy}b_y 1.
    \end{align*}
    To make sense of this, we recognize the second piece as the inner product $\frac{1}{2}m^{\dagger}\nabla S^{-1}m$, since $b_x 1 = m_x$ by definition.
    We can use the commutator to evaluate the first piece, and we find
    \begin{align*}
        \frac{1}{2}\nabla S^{-1}_{xy}c_xb_y 1 &= \frac{1}{2}\nabla S^{-1}_{xy}D_{xy} \\
        &=
        \frac{1}{2}\mathrm{tr}\left(\nabla S^{-1}D\right).
    \end{align*}
    Putting the pieces together and integrating, we get
    $$
    \frac{1}{2} \Phi_x\nabla S^{-1}_{xy}\Phi_y 1 = \frac{1}{2}m^{\dagger}\nabla S^{-1}m + \frac{1}{2}\mathrm{tr}(\nabla S^{-1} D).
    $$
\end{proof}
Inserting the prior and posterior mean and covariance into the expression derived in Lemma~\ref{lem:C1}, we find the respective expectations.
For the prior, we have
\begin{equation}
    \label{eqn:priorexp}
    \mathbb{E}\left[\nabla H(\phi;\lambda)\middle|\lambda\right] = \frac{1}{2}\mathrm{tr}\left(\nabla S^{-1}S\right),
\end{equation}
and for the posterior,
\begin{equation}
    \label{eqn:postexp}
    \mathbb{E}\left[\nabla H(\phi;\lambda)\middle|d, \lambda\right] = \frac{1}{2}\Tilde{m}^{\dagger}\nabla S^{-1}\Tilde{m} + \frac{1}{2}\mathrm{tr}\left(\nabla S^{-1}\Tilde{S}\right).
\end{equation}

Next, we do the same trick for the expectation of $\nabla H(\phi;\lambda)^T \nabla H(\phi;\lambda)$, taken over the prior and posterior.
Using the operator calculus, we will evaluate $\mathbb{E}\left[ \nabla H(\phi;\lambda)^T \nabla H(\phi;\lambda) \middle| \mathcal{N}(m, D) \right]$ for a general $m$ and $D$, and insert the prior/posterior mean and covariance at the end.
We summarize the result in the following Lemma.
\newcommand{\intd}{\frac{1}{4}\int dx\: dy\: dz\: dw\:}
\newcommand{\ns}{\nabla S^{-1}}

\begin{lemma}
    \label{lem:C2}
    The expectation of $\nabla H(\phi;\lambda)^T\nabla H(\phi;\lambda)$, taken over $\mathcal{N}(m,D)$, is given by
    \begin{align*}
        \mathbb{E} \left[ \nabla H(\phi;\lambda)^T \nabla H(\phi;\lambda) \middle |\mathcal{N}(m, D)\right] &= \frac{1}{4}\left(m^{\dagger} \ns m\right)^2 + \frac{1}{2} \left(m^{\dagger} \ns m\right) \mathrm{tr}\left(\ns D\right) \\
        &+ m^{\dagger} \ns D \ns m + \frac{1}{2}\mathrm{tr}\left(\ns D \ns D\right) + \frac{1}{4} \mathrm{tr}^2\left(\ns D\right).
\end{align*}
\end{lemma}

\begin{proof}
    We recognize $F(\phi) = \nabla H(\phi;\lambda)^T \nabla H(\phi;\lambda)$, place $\Phi$ in place of $\phi$, and let it act on $1$, so we see we must evaluate
    \begin{equation}
        \label{eqn:h2int}
        F(\Phi)1 = \frac{1}{4} \Phi_x \ns_{xy} \Phi_y \Phi_z \ns_{zw} \Phi_w 1 = \frac{1}{4} \ns_{xy} \ns_{zw} \Phi_x \Phi_y \Phi_z  \Phi_w 1,
    \end{equation}
    since the field-operators commute with $\ns$.
    To evaluate eq.~(\ref{eqn:h2int}), we begin by expressing $\Phi$ with the creation and annihilation operators and work out the result.
    We let $F(\Phi)$ act on $1$, and we find:
    \begin{align*}
        F(\Phi)1 &= \frac{1}{4} (c_x + b_x) \ns_{xy} (c_y + b_y)(c_z + b_z) \ns_{zw} (c_w + b_w)1 \\
        &= \frac{1}{4} (c_x + b_x) \ns_{xy} (c_y + b_y)(c_z + b_z) \ns_{zw} b_w1 \\
        &= \frac{1}{4} (c_x + b_x) \ns_{xy} (c_y + b_y)b_z \ns_{zw} b_w1 + \frac{1}{4} (c_x + b_x) \ns_{xy} (c_y + b_y) \ns_{zw} c_zb_w1 \\
        &= \frac{1}{4} (c_x + b_x) \ns_{xy} c_y b_z \ns_{zw} b_w1 + \frac{1}{4} (c_x + b_x) \ns_{xy} b_y b_z \ns_{zw} b_w1 \\
        &+ \frac{1}{4} (c_x + b_x) \ns_{xy} c_y \ns_{zw} c_z b_w1 + \frac{1}{4} (c_x + b_x) \ns_{xy} b_y \ns_{zw} c_z b_w1 \\
        &= \frac{1}{4} c_x \ns_{xy} c_y b_z \ns_{zw} b_w1 + \frac{1}{4} b_x \ns_{xy}c_y b_z \ns_{zw} b_w1 \\
        &+ \frac{1}{4} c_x \ns_{xy} b_y b_z \ns_{zw} b_w1 + \frac{1}{4} b_x \ns_{xy} b_y b_z \ns_{zw} b_w1 \\
        &+ \frac{1}{4} c_x \ns_{xy} c_y \ns_{zw} c_z b_w1 + \frac{1}{4} b_x \ns_{xy} c_y \ns_{zw} c_z b_w1 \\
        &+ \frac{1}{4} c_x \ns_{xy} b_y \ns_{zw} c_z b_w1 + \frac{1}{4} b_x \ns_{xy} b_y \ns_{zw} c_z b_w1,
    \end{align*}
    and we will evaluate of each of these one by one.

    The first term gives:
    \begin{align*}
        c_x \ns_{xy} c_y b_z \ns_{zw} b_w1 &= \ns_{xy} c_x c_y b_z \ns_{zw} b_w1 \\
        &= \ns_{xy} c_x ([c_y,b_z] + b_z c_y) \ns_{zw} b_w1 \\
        &= \ns_{xy} c_x D_{yz} \ns_{zw} b_w1 + \ns_{xy} c_x b_z c_y \ns_{zw} b_w1 \\
        &= \ns_{xy} D_{yz} \ns_{zw} c_x b_w1 + \ns_{xy} c_x b_z \ns_{zw} c_y b_w1 \\
        &= \ns_{xy} D_{yz} \ns_{zw} (D_{xw} + b_w c_x)1 + \ns_{xy} c_x b_z \ns_{zw} (D_{yw} + b_w c_y)1 \\
        &= \ns_{xy} D_{yz} \ns_{zw} D_{xw}1 + \ns_{xy} c_x b_z \ns_{zw} D_{yw}1 \\
        &= \ns_{xy} D_{yz} \ns_{zw} D_{wx}1 + \ns_{xy} (D_{xz} + b_z c_x) \ns_{zw} D_{yw}1 \\
        &= \ns_{xy} D_{yz} \ns_{zw} D_{wx}1 + \ns_{xy} D_{xz} \ns_{zw} D_{yw}1 + \ns_{xy} b_z c_x \ns_{zw} D_{yw}1 \\
        &= \ns_{xy} D_{yz} \ns_{zw} D_{wx}1 + \ns_{yx} D_{xz} \ns_{zw} D_{wy}1 + \ns_{xy} b_z \ns_{zw} D_{yw} c_x 1 \\
        &= 
        2 \ns_{xy} D_{yz} \ns_{zw} D_{wx}1.
    \end{align*}
    After doing the integral, we find:
    \begin{equation*}
        \label{eqn:1var}
        \frac{1}{4} c_x \ns_{xy} c_y b_z \ns_{zw} b_w1 = \frac{1}{2} \mathrm{tr}(\ns D \ns D).
    \end{equation*}
    Next, we have
    \begin{align*}
        b_x \ns_{xy}c_y b_z \ns_{zw} b_w1 &= b_x \ns_{xy} (D_{yz} + b_z c_y) \ns_{zw} b_w1 \\
        &= b_x \ns_{xy} D_{yz} \ns_{zw} b_w 1+ b_x \ns_{xy} b_z \ns_{zw} D_{yw}1 \\
        &= b_x \ns_{xy} D_{yz} \ns_{zw} b_w1 + b_x \ns_{xy} D_{yw} \ns_{wz} b_z1 \\
        &= 2 b_x \ns_{xy} D_{yz} \ns_{zw} b_w1,
    \end{align*}
    so
    \begin{equation*}
        \label{eqn:2var}
        \frac{1}{4} b_x \ns_{xy}c_y b_z \ns_{zw} b_w 1= \frac{1}{2} m^{\dagger} \ns D \ns m.
    \end{equation*}
    The third term is
    \begin{align*}
        c_x \ns_{xy} b_y b_z \ns_{zw} b_w1 &= \ns_{xy} c_x b_y b_z \ns_{zw} b_w1 \\
        &= \ns_{xy} (D_{yx} + b_y c_x) b_z \ns_{zw} b_w1 \\
        &= \ns_{xy} D_{yx} b_z \ns_{zw} b_w 1+ \ns_{xy} b_y c_x b_z \ns_{zw} b_w 1\\
        &= \ns_{xy} D_{yx} b_z \ns_{zw} b_w1 + \ns_{xy} b_y (D_{xz} + b_z c_x) \ns_{zw} b_w 1\\
        &= \ns_{xy} D_{yx} b_z \ns_{zw} b_w1 + \ns_{xy} b_y D_{xz} \ns_{zw} b_w1 + b_y \ns_{yx} D_{xw} \ns_{wz} b_z1 \\
        &= \ns_{xy} D_{yx} b_z \ns_{zw} b_w 1+ 2 b_y \ns_{yx} D_{xw} \ns_{wz} b_z1,
    \end{align*}
    leading to
    \begin{equation*}
        \label{eqn:3var}
        \frac{1}{4} c_x \ns_{xy} b_y b_z \ns_{zw} b_w1 = \frac{1}{4} m^{\dagger}\ns m \mathrm{tr}(\ns D) + \frac{1}{2} m^{\dagger} \ns D \ns m.
    \end{equation*}
    The fourth term is of a squared, quadratic form:
    \begin{equation*}
        \label{eqn:4var}
        \frac{1}{4} b_x \ns_{xy} b_y b_z \ns_{zw} b_w 1= \frac{1}{4}\left(m^{\dagger} \ns m \right)^2.
    \end{equation*}
    We find the fifth and sixth terms to be zero
    \begin{align*}
        c_x \ns_{xy} c_y \ns_{zw} c_z b_w1 &= c_x \ns_{xy} c_y \ns_{zw} D_{zw}1 \\
        &= c_x \ns_{xy} \ns_{zw} D_{zw} c_y 1 = 0
    \end{align*}
    and
    \begin{align*}
        b_x \ns_{xy} c_y \ns_{zw} c_z b_w1 &= b_x \ns_{xy} c_y \ns_{zw} D_{zw}1 \\
        &= b_x \ns_{xy} \ns_{zw} D_{zw} c_y 1 = 0.
    \end{align*}
    The seventh term is:
    \begin{align*}
        c_x \ns_{xy} b_y \ns_{zw} c_z b_w1 &= \ns_{xy} c_x b_y \ns_{zw} D_{zw}1 \\
        &= \ns_{xy}(D_{xy} + b_y c_x) \ns_{zw} D_{zw}1 \\
        &= \ns_{xy} D_{xy} \ns_{zw} D_{zw}1 + \ns_{xy} b_y c_x \ns_{zw} D_{zw}1 \\
        &= \ns_{xy} D_{yx} \ns_{zw} D_{wz}1 + \ns_{xy} b_y \ns_{zw} D_{wz} c_x 1 \\
        &= \ns_{xy} D_{yx} \ns_{zw} D_{wz}1,
    \end{align*}
    and integrating we find:
    \begin{equation*}
        \label{eqn:7var}
        \frac{1}{4} c_x \ns_{xy} b_y \ns_{zw} c_z b_w 1= \frac{1}{4} \mathrm{tr}^2\left(\ns D\right).
    \end{equation*}
    Finally, the last term is:
    \begin{align*}
        b_x \ns_{xy} b_y \ns_{zw} c_z b_w1 &= b_x \ns_{xy} b_y \ns_{zw}(D_{wz} + b_w c_z)1 \\
        &= b_x \ns_{xy} b_y \ns_{zw} D_{wz}1,
    \end{align*}
    and when we integrate, we get:
    \begin{equation*}
        \label{eqn:8var}
        \frac{1}{4} b_x \ns_{xy} b_y \ns_{zw} c_z b_w 1= \frac{1}{4} m^{\dagger}\ns m \mathrm{tr}(\ns D).
    \end{equation*}
\end{proof}
Thus, the desired expectations of $\nabla H(\phi;\lambda)^T \nabla H(\phi;\lambda)$ taken over the prior and posterior in the free theory case can be found by replacing $m$ and $D$ with the appropriate mean and covariance for each in the expression derived in Lemma~\ref{lem:C2}.

Doing so, we find for the prior
\begin{equation}
    \label{eqn:squaredpriorexp}
    \mathbb{E}\left[\nabla H(\phi;\lambda)^T \nabla H(\phi;\lambda) \middle|\lambda\right] = \frac{1}{2}\mathrm{tr}\left(\ns S \ns S\right) + \frac{1}{4} \mathrm{tr}^2\left(\ns S\right),
\end{equation}
and the posterior gives:
\begin{multline}
    \label{eqn:squaredpostexp}
    \mathbb{E}\left[\nabla H(\phi;\lambda)^T \nabla H(\phi;\lambda) \middle|d, \lambda\right] = \frac{1}{4}\left(\Tilde{m}^{\dagger} \ns \Tilde{m}\right)^2 + \frac{1}{2} \left(\Tilde{m}^{\dagger} \ns \Tilde{m}\right) \mathrm{tr}\left(\ns \Tilde{S}\right) \\
    + \Tilde{m}^{\dagger} \ns \Tilde{S} \ns \Tilde{m} + \frac{1}{2}\mathrm{tr}\left(\ns \Tilde{S} \ns \Tilde{S}\right) + \frac{1}{4} \mathrm{tr}^2\left(\ns \Tilde{S}\right).
\end{multline}

Finally, the covariance of $\nabla H(\phi;\lambda)$ taken over the prior in eq.~(\ref{eqn:priorcov}) can be derived with equations eq.~(\ref{eqn:squaredpriorexp}) and eq.~(\ref{eqn:priorexp}),
\begin{equation*}
    \mathbb{C}\left[\nabla_{\lambda}H(\phi;\lambda)|\lambda\right] = \frac{1}{2} \mathrm{tr}\left( \nabla_{\lambda} S^{-1}_{\lambda} S_{\lambda} \nabla_{\lambda} S^{-1}_{\lambda} S_{\lambda} \right).
\end{equation*}
Likewise, the covariance of $\nabla H(\phi;\lambda)$ taken over the posterior, as given in eq.~(\ref{eqn:postcov}), can be found using equations eq.~(\ref{eqn:squaredpostexp}) and eq.~(\ref{eqn:postexp})
\begin{equation*}
    \mathbb{C}\left[\nabla_{\lambda}H(\phi;\lambda)|d, \lambda\right] =\Tilde{m}_{\lambda}^{\dagger} \nabla_{\lambda} S_{\lambda}^{-1} \Tilde{S}_{\lambda} \nabla_{\lambda} S_{\lambda}^{-1} \Tilde{m} + \frac{1}{2} \mathrm{tr}\left( \nabla_{\lambda} S^{-1}_{\lambda} \Tilde{S}_{\lambda} \nabla_{\lambda} S^{-1}_{\lambda} \Tilde{S}_{\lambda} \right),
\end{equation*}
which completes the proof.

\end{document}